\documentclass[a4paper,10pt,openany]{article}
\usepackage{etex}
\usepackage[utf8x]{inputenc}
\usepackage[english]{babel}
\usepackage{amsmath,amssymb,amsthm,mathrsfs,amsfonts,dsfont}
\usepackage{graphicx}
\usepackage[footnotesize]{caption}
\usepackage{epigraph}
\usepackage{indentfirst}
\usepackage{booktabs}
\usepackage{fancyhdr}
\usepackage{vmargin}
\usepackage{feynmf}
\usepackage{yfonts}
\usepackage{lscape}
\usepackage{subfig}
\usepackage{textcomp}
\usepackage{pstricks}
\usepackage[metapost]{mfpic}
\usepackage{fancybox}
\usepackage{slashed}
\usepackage[verbose]{wrapfig}
\usepackage{pifont}
\usepackage{keystroke}
\usepackage{appendix}
\usepackage{enumitem}
\usepackage{dyntree}
\usepackage{array}
\usepackage{colortbl}
\usepackage{alltt}
\usepackage{boxedminipage}
\usepackage{tikz}
\usepackage{calc}
\usepackage{subfloat}
\usepackage{multirow}
\usepackage{cmll}
\usepackage{multicol}
\definecolor{color1}{RGB}{204,0,51}
\definecolor{color2}{RGB}{159,182,205}
\usepackage{bookmark}

\usepackage{verbatim}
\usepackage[vcentermath]{youngtab}
\usepackage{young}
\usepackage{pst-grad} 
\usepackage{pst-plot} 
\usepackage{transparent}
\usepackage{color,soul}
\usepackage{cancel}
\usepackage{bm}
\usepackage{placeins}
\usepackage{mathtools}
\usepackage{rotating}
\usepackage[refpage]{nomencl}
\usepackage{hhline}
\usepackage{marginnote}
\usepackage{upgreek}
\usepackage{stackengine}
\usepackage{setspace}
\usepackage{ytableau}
\usepackage{listings}  
\usepackage{calligra}
\usepackage{hyperref}
\usepackage{afterpage}
\usepackage{pagecolor}
\usepackage{cite}
\usepackage{cleveref}
\allowdisplaybreaks
\lstdefinestyle{customc}{
  belowcaptionskip=1\baselineskip,
  breaklines=true,
  frame=L,
  xleftmargin=\parindent,
  language=C,
  showstringspaces=false,
  basicstyle=\footnotesize\ttfamily,
  keywordstyle=\bfseries\color{green!40!black},
  commentstyle=\itshape\color{purple!40!black},
  identifierstyle=\color{blue},
  stringstyle=\color{red217!80!black},
}

\lstdefinestyle{customasm}{
  belowcaptionskip=1\baselineskip,
  frame=L,
  xleftmargin=\parindent,
  language=[x86masm]Assembler,
  basicstyle=\footnotesize\ttfamily,
  commentstyle=\itshape\color{purple!40!black},
}

\sethlcolor{yellow}
\definecolor{red217}{RGB}{217,0,0}
\definecolor{boh}{RGB}{79,47,79}

\hypersetup{
linkbordercolor={red}
}

\hypersetup{
    bookmarks=false,         
    unicode=false,          
    pdftoolbar=true,        
    pdfmenubar=true,        
    pdffitwindow=false,     
    pdfstartview={FitH},    
    pdftitle={},    
    pdfauthor={},     
    pdfsubject={},   
    pdfnewwindow=true,      
    colorlinks=false,       
    linkcolor=red,          
    citecolor=red,        
    filecolor=red,      
    urlcolor=red           
    linkbordercolor={red},
    citebordercolor={red},
    urlbordercolor={red}
}

\usepackage{listings} 

\makeatletter

\newcommand{\Rmnum}[1]{\expandafter\@slowromancap\romannumeral #1@}
\makeatother

\newcommand{\scp}[2]{\langle #1, #2 \rangle}

\newtheorem{thm}{Theorem}[section]

\theoremstyle{remark}

\theoremstyle{definition}
\newtheorem{definition}[thm]{Definition}

\theoremstyle{proposition}
\newtheorem{proposition}[thm]{Proposition}

\theoremstyle{Lemma}

\usepackage{color}

\def\chapterautorefname~#1\null{chap.~(#1)\null}
\def\sectionautorefname~#1\null{sec.~(#1)\null}
\def\subsectionautorefname~#1\null{sub--Sec.~(#1)\null}
\def\figureautorefname~#1\null{fig.~(#1)\null}
\def\tableautorefname~#1\null{tab.~(#1)\null}
\def\equationautorefname~#1\null{eq.~(#1)\null}

\renewcommand{\Re}{\operatorname{Re}}
\renewcommand{\Im}{\operatorname{Im}}

\newcommand{\bl}[1]{\textcolor{blue}{#1}}

\begin{document}
\def\chapterautorefname~#1\null{chap.~(#1)\null}
\def\sectionautorefname~#1\null{sec.~(#1)\null}
\def\subsectionautorefname~#1\null{sub--Sec.~(#1)\null}
\def\figureautorefname~#1\null{fig.~(#1)\null}
\def\tableautorefname~#1\null{tab.~(#1)\null}
\def\equationautorefname~#1\null{eq.~(#1)\null}

\begin{titlepage}
\begin{center}


\vskip 1.5cm

{\Large \bf Branes and Polytopes } 

\vskip 1cm

{\bf Luca Romano}

\vskip 25pt


{email address: {\tt lucaromano2607@gmail.com } \\}

\end{center}

\vskip 0.5cm

\begin{center} {\bf ABSTRACT}\\[3ex]
\end{center}

We investigate the hierarchies of half-supersymmetric branes in maximal supergravity theories. By studying the action of the Weyl group of the U-duality group of maximal supergravities we discover a set of universal algebraic rules describing the number of independent 1/2-BPS p-branes, rank by rank, in any dimension. 
We show that these relations describe the symmetries of certain families of uniform polytopes. This induces a correspondence between half-supersymmetric branes and vertices of opportune uniform polytopes.  We show that half-supersymmetric 0-, 1- and 2-branes are in correspondence with the vertices of the $k_{21}$, $2_{k1}$ and $1_{k2}$ families of uniform polytopes, respectively, while 3-branes correspond to the vertices of the rectified version of the $2_{k1}$ family. 
For 4-branes and higher rank solutions we find a general behavior. The interpretation of half-supersymmetric solutions as vertices of uniform polytopes reveals some intriguing aspects. One of the most relevant is a triality relation between 0-, 1- and 2-branes. 

\end{titlepage}

\newpage \setcounter{page}{1}

\tableofcontents

\newpage

\section*{Introduction}\addcontentsline{toc}{section}{Introduction} 


Since their discovery branes gained a prominent role in the analysis of M-theories and dualities \cite{1987PhLB..189...75B}.  One of the most important class of branes consists in Dirichlet branes, or D-branes. D-branes appear in string theory as boundary terms for open strings with mixed Dirichlet-Neumann boundary conditions and, due to their tension, scaling with a negative power of the string coupling constant, they are non-perturbative objects\cite{Polchinski:1995mt}. Moreover D-branes were also used in the derivation of the black hole entropy by a microstate counting, one of the most relevant results of string theory \cite{Strominger:1996sh}. In type II string theory the coupling of D-branes with the Ramond-Ramond (RR) sector is described by the Wess-Zumino term.  
This framework could be found in the low-energy limit of string theory, supergravity, where branes occur as classical solutions coupled to differential forms \cite{Witten:1995ex}. For these reasons brane solutions in supergravity could be used as a probe to take a look inside the non perturbative regime of string theory and to improve our knowledge of dualities\cite{Townsend:1995kk,Bergshoeff:1996tu}. Branes play also a relevant role in cosmological models, as the  brane-world scenario and in the AdS/CFT correspondence \cite{Randall:1999ee, Randall:1999vf}.\\

The U-duality group of M-theory emerges in supergravity, in its continuous version, $E_{11-d}(\mathds{R})$, as a global symmetry \cite{Obers:1998fb, Julia:1997cy,Hull:1995mz} and the differential (p+1)-form potentials coupling with p-brane solutions belong to representations of this group. For this reason many attempts to investigate branes solutions in supergravity are based on the algebraic structure provided by the U-duality group\cite{Bergshoeff:2010xc, Bergshoeff:2011qk, Bergshoeff:2012ex, Kleinschmidt:2011vu}. One of the most relevant achievements in this field comes from the classification of the invariants of the U-duality group and the corresponding orbits, directly linked to physical relevant quantity, as entropy \cite{Ferrara:1997ci, Ferrara:1997uz, Lu:1997bg}. In the context of branes a special role is played by half-supersymmetric solutions. These preserve the maximum amount of supersymmetry and could be considered as building blocks for less supersymmetric states, that could be realized as bound states of them.\\
%
%

Depending on the number of spatial transverse directions branes could be divided in two classes, standard branes and non-standard branes; the former have three or more transverse spatial directions, the latter two or less. Physically the number of transverse directions characterizes their asymptotic behavior and, while standard branes approach flat Minkowsky, this is not true for non-standard branes. In the class of non-standard branes we recognize defect branes, domain walls and spacefilling branes corresponding respectively to $(d-3)$-, $(d-2)$- and $(d-1)$-branes in $d$ dimensions. Although single states of these branes have infinite energy, finite energy solutions could be realized as a bound states of them in presence of an orientifold. Defect branes couple to $(d-2)$-forms that are dual to scalars. Domain walls and spacefilling branes couple to $(d-1)$- and $d$-forms; despite these do not carry any degrees of freedom domain walls and spacefilling branes play a relevant role in different contexts\cite{Bergshoeff:2012pm}. As a first step towards a taxonomy of half-supersymmetric branes in supergravity the classification of the differential form potentials is crucial. A full classification was completed in the IIA and  IIB supergravity theories by requiring the closure of the supersymmetry and gauge algebras \cite{Bergshoeff:2005ac,Bergshoeff:2006qw,Bergshoeff:2010mv}.
This approach could be be generalized to all supergravity theories by the $E_{11}$ construction \cite{Riccioni:2007au,West:2001as}. The very extended Kac-Moody algebra $E_{11}$ contains, for any maximal supergravity, both the spacetime symmetry and the U-duality algebra $E_{11-d}$. The spectrum of differential forms could be obtained by decomposing the adjoint representation of $E_{11}$ in its subgroup $E_{11-d}\times GL(d,\mathds{R})$, where 
the two factors are the U-duality group of the $d$ dimensional maximal theory  and the spacetime symmetry respectively, and selecting the real states, identified by a positive squared norm.\\

1/2-BPS branes in maximal theories have been characterized from a pure group-theoretical point of view by showing that they couple to differential form potentials corresponding to the longest weights of the U-duality representations they belong to\cite{2011JHEP...10..144K, Bergshoeff:2013sxa}. This classification points out a consistent difference between standard and non-standard solutions. Indeed, non-standard branes belong to representations with a nontrivial length stratification,  while standard branes always live in representations without any length stratification. This implies all the components of the differential form potentials couple to half-supersymmetric solutions in the case of standard branes, while, for non-standard brane solutions, only a subset of them couple to half-supersymmetric solutions. This behavior reflects the possibility to combine non-standard brane solutions in a bound state  preserving the same amount of supersymmetry of the single branes. i.e. there is a degeneracy with respect to the BPS condition. \\

Opposite to maximal theories in non-maximal supergravities the U-duality group does not appear, in general, in its maximal non-compact form. The presence of compact and non-compact weights requires a careful analysis in extending the previous correspondence. Using the theory of real forms of Lie algebras and Tits-Satake diagrams \cite{araki} it has been argued that half-maximal solutions couple only to non-compact longest weights \cite{Bergshoeff:2014lxa,Marrani:2015gfa}. The refined rule reproduces the previous results when applied to the maximal case, where the split form for the U-duality group prevents from the presence of compact weights. In this picture the Weyl group of the U-duality group plays a remarkable role, since it maps solutions to solutions preserving their supersymmetric amount \cite{Lu:1996ge}.\\

Although the correspondence between longest non-compact weights and half-supersymmetric solutions provides us with an elegant algebraic characterization for 1/2-BPS branes in supergravity theories, we believed there was still a lack of a global view of the network of these solutions. In particular we argued that the role of the Weyl group associated with the U-duality group was not yet fully used to investigate the presence of a universal structure behind 1/2-BPS solutions. In order to uncover the algebraic structure governing half-supersymmetric branes in maximal theories we applied the general theory of reflection groups and Coxeter groups using, as starting point, the correspondence between longest non-compact weights and branes. We discovered a set of algebraic rules describing the content of 1/2-BPS branes in maximal theories, rank by rank, in any dimension.  Moreover, the interpretation of these rules as symmetries of certain families of uniform polytopes, induces a correspondence between branes and polytopes. Half-supersymmetric solutions in maximal theories could be seen as vertices of opportune uniform polytopes. We believe this link could provide consistent improvements in understanding duality relations and connections between different brane solutions.\\

The paper is organized as follows. In the first section the general theory of Coxeter groups and reflection groups is reviewed, providing the basic tools needed in our investigation. In \cref{sec:braneE11Weyl} we introduce the $E_{11}$ construction deriving all the representations hosting differential forms in maximal theories from three to nine dimensions.  We also discuss the role of the Weyl group associated with the U-duality group. In \cref{sec:3} the first part of our original work is exposed; we apply some general results concerning Coxeter group to maximal theories. In particular we study the orbits of the highest weights of the U-duality representations under the Weyl action. This leads us to a set of algebraic rules capturing the algebraic structure behind half-supersymmetric solutions. \Cref{sec:branesandpolytopes} is devoted to the interpretation of these rules as symmetries of uniform polytopes. We recall the basic tools to deal with polytopes and their relation with Coxeter groups. We recognize that half-supersymmetric 0-, 1- and 2-branes could be thought as vertices of the families of uniform polytopes $k_{21},2_{k1}$ and $1_{k2}$ respectively. This correspondence reveals a triality relations between these solutions. By the same way we discuss the correspondence for higher rank solutions.  In the conclusions we summarize our work and point out possible outlooks. In \cref{sec:appendix1} we list all the features of the uniform polytopes involved in our analysis, while in \cref{sec:appB} we give a basic introduction to Petrie polygons.


%

\section{Coxeter Group and Weyl Group}\label{sec:1}
In this section we give a brief introduction to reflections groups, Coxeter groups and Weyl groups. We begin with the definition of Coxeter group \cite{Coxeter1940,humphreys1990, 10.2307/1968753}
\begin{definition}[Coxeter Group]
 Given  a set of generators $S=\{r_1,...,r_{n}\}$ a \textbf{Coxeter group} $W$ is the group generated by $S$ with presentation
 \begin{flalign}
  &\langle\ r_{1},r_{2},...,r_{n}|\ (r_{i}r_{j})^{m_{ij}}=\mathds{1}\ \rangle,\label{Coxeterrelations1}
 \end{flalign}
where $m_{ij}\in \mathds{Z}\cup\{\infty\}$, $m_{ii}=1$ and $m_{ij}\geqslant 2$ for $i\neq j$. 
\end{definition}
$m_{ij}$ is the order of the element product $r_{i}r_{j}$. If $m_{ij}=\infty$ it means no relation of the form above could be imposed on $r_{i}$ and $r_{j}$. $r_{i}$ are often referred as \textbf{simple reflections}. $m_{ii}=1$  imply that all the simple reflections are involutions. Two simple reflections $r_{i},\ r_{j}$ commute if their product has order 2, $m_{ij}=2$. Furthermore by $m_{ii}=1$, if $(r_{i}r_{j})^{m_{ij}}=1$ it follows
\begin{flalign}
&(r_{j}r_{i})^{m_{ij}}=r_{i}r_{i}(r_{j}r_{i})^{m_{ij}}=r_{i}(r_{i}r_{j})^{m_{ij}}r_{i}=1,
\end{flalign}
thus we assume $m_{ij}=m_{ji}$. If $W$ is a Coxeter group and $S=\{r_1,...,r_{n}\}$ the set of its generators, the pair $(W, S)$ is called \textbf{Coxeter system}. The number of generators is the rank of the Coxeter system. The values of $m_{ij}$ for any Coxeter system could be collected in a symmetric matrix $M$ with entries in $\mathds{Z}\cup\{\infty\}$,
\begin{flalign}
 &M_{ij}=m_{ij}
\end{flalign}
called \textbf{Coxeter matrix}. Another relevant matrix associated with a Coxeter system is the \textbf{Schl\"afli matrix} whose entries are defined by
\begin{flalign}
 &C_{ij}=-2\cos\left(\frac{\pi}{m_{ij}}\right).\label{schlafli1}
\end{flalign}
Any Coxeter group could be described by a graph, the \textbf{Coxeter graph}, in a way similar to the description of Lie algebras by means of Dynkin diagrams. In particular, given a Coxeter system $(W,S)$, its associated Coxeter graph is the undirected graph drawn with the following prescriptions
\begin{enumerate}[label=(\roman*)]
  \item Any generator corresponds to a vertex in the graph.
  \item Vertices corresponding to the generators $r_{i}$ and $r_{j}$ are connected by an edge if $m_{ij}\geqslant 3$.
  \item Edges are labeled with the value of $m_{ij}$; if $m_{ij}=3$ the label could be omitted.
\end{enumerate}
A Coxeter system $(W,S)$ is said to be \textbf{irreducible} if its graph is connected. Its is immediate to recover the Coxeter matrix and Schl\"afli matrix from a Coxeter graph \cite{Coxeter1940}. As an example, taking the graph, in \cref{fig1}
\begin{figure}[h!]
\centering
\scalebox{0.6} 
{
\begin{pspicture}(0,-0.44)(2.8,0.4)
\definecolor{color816b}{rgb}{0.00392156862745098,0.00392156862745098,0.00392156862745098}
\psline[linewidth=0.04cm](0.2,0.2)(2.6,0.2)
\pscircle[linewidth=0.02,dimen=outer,fillstyle=solid,fillcolor=color816b](0.2,0.2){0.2}
\pscircle[linewidth=0.02,dimen=outer,fillstyle=solid,fillcolor=color816b](2.6,0.2){0.2}
\pscircle[linewidth=0.02,dimen=outer,fillstyle=solid,fillcolor=color816b](1.4,0.2){0.2}
\usefont{T1}{ptm}{m}{n}
\rput(0.7275,-0.255){\Large 4}
\end{pspicture} 
}

\caption{An example of Coxeter graph.}\label{fig1}
\end{figure}
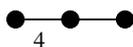\\ 
one finds
\begin{align}
&M=\left(\begin{array}{ccc}
         1&4&2\\
         4&1&3\\
         2&3&1
        \end{array}
\right)&&C=
\left(\begin{array}{ccc}
         2&-\sqrt{2}&0\\
         -\sqrt{2}&2&-1\\
         0&-1&2
        \end{array}
\right)
\end{align}
According to the eigenvalues of its Schl\"afli matrix a Coxeter system is classified in 
\begin{enumerate}[label=(\roman*)]
 \item \textbf{Finite type} if the Schl\"afli matrix is positive definite, namely it has all positive eigenvalues.
 \item \textbf{Affine type} if the Schl\"afli matrix is semipositive definite, namely it has all non-negative eigenvalues.
 \item \textbf{Indefinite type} otherwise.
\end{enumerate}
\textbf{Hyperbolic type} Coxeter groups belong to the irreducible indefinite type with the further condition that any proper connected subgraph of its Coxeter graph describes a Coxeter system either of finite or affine type.\\

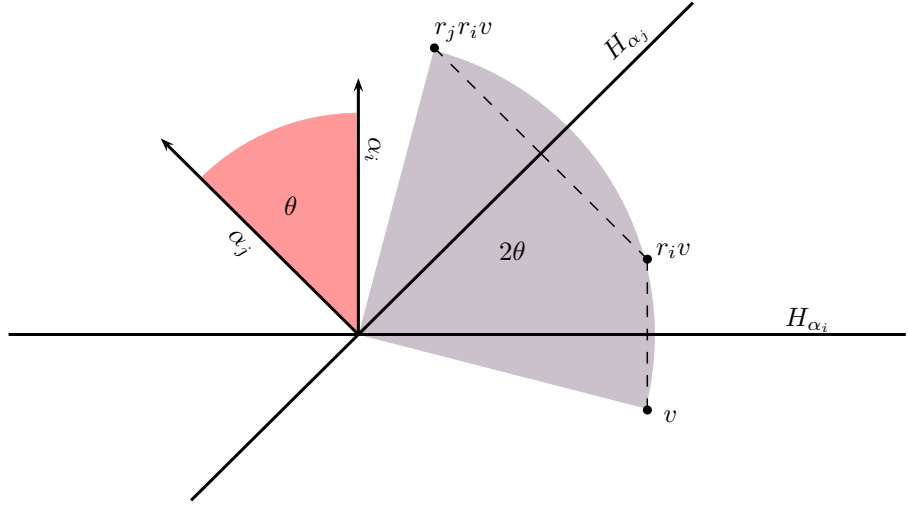
\begin{figure}[h!]
\centering
\scalebox{1} 
{
\begin{pspicture}(0,-2.8030114)(12.119062,4.5712957)
\definecolor{color529b}{rgb}{1.0,0.8431372549019608,0.0}
\definecolor{color529}{rgb}{0.996078431372549,0.996078431372549,0.996078431372549}
\definecolor{color760b}{rgb}{0.8941176470588236,0.00784313725490196,0.027450980392156862}
\pswedge[linewidth=0.02,linecolor=color529,fillstyle=solid,fillcolor=boh!30](4.59,-0.5930115){3.93}{-14.534454}{75.29934}
\pswedge[linewidth=0.02,linecolor=color529,fillstyle=solid,fillcolor=red!40](4.61,-0.5930115){2.97}{90.47351}{135.0}
\usefont{T1}{ppl}{m}{n}
\rput{-90.0}(2.943324,6.6057386){\rput(4.744531,1.8469886){$\alpha_{i}$}}
\usefont{T1}{ppl}{m}{n}
\rput{-45.0}(0.48488915,2.3130405){\rput(3.0045311,0.5869885){$\alpha_{j}$}}
\psdots[dotsize=0.12](8.4,-1.5830115)
\psline[linewidth=0.02cm,linestyle=dashed,dash=0.16cm 0.16cm](8.4,-1.5830115)(8.4,0.4169885)
\psdots[dotsize=0.12](8.4,0.4169885)
\usefont{T1}{ppl}{m}{n}
\rput(8.704532,-1.6730115){$v$}
\usefont{T1}{ppl}{m}{n}
\rput(8.744532,0.5269885){$r_{i}v$}
\psline[linewidth=0.02cm,linestyle=dashed,dash=0.16cm 0.16cm](8.4,0.4169885)(7.0,1.8169885)
\psline[linewidth=0.02cm,linestyle=dashed,dash=0.16cm 0.16cm](7.0,1.8169885)(5.6,3.2169886)
\usefont{T1}{ppl}{m}{n}
\rput(5.9845314,3.4269886){$r_{j}r_{i}v$}
\usefont{T1}{ppl}{m}{n}
\rput(3.7145312,1.1269885){$\theta$}
\usefont{T1}{ppl}{m}{n}
\rput(6.6245313,0.5269885){$2\theta$}
\usefont{T1}{ppl}{m}{n}
\rput(10.514531,-0.3930115){$H_{\alpha_{i}}$}
\usefont{T1}{ppl}{m}{n}
\rput{45.0}(4.732713,-4.8033657){\rput(8.134531,3.3269885){$H_{\alpha_{j}}$}}
\psdots[dotsize=0.12](5.6,3.2169886)
\psline[linewidth=0.04cm](0.0,-0.5830115)(11.8,-0.5830115)
\psline[linewidth=0.04cm,arrowsize=0.05291667cm 2.0,arrowlength=1.4,arrowinset=0.4]{<-}(4.6,2.8169885)(4.6,-0.5830115)
\psline[linewidth=0.04cm](2.4,-2.7830114)(9.0,3.8169885)
\psline[linewidth=0.04cm,arrowsize=0.05291667cm 2.0,arrowlength=1.4,arrowinset=0.4]{<-}(2.0,2.0169885)(4.6,-0.5830115)
\end{pspicture} 
}
\caption{A sequence of two reflections with respect to two planes, $H_{\alpha_i}$ and $H_{\alpha_j}$, at angle $\theta$ corresponds to a rotation of $2\theta$ in the plane spanned by $\alpha_{i}$ and $\alpha_{j}$.} \label{rotation1}

\end{figure}

Now we spend some words on the geometric interpretation of a Coxeter system. (W,S) could be realized geometrically as the group generated by orthogonal reflections on a vector space V over $\mathds{R}$. In particular taking a basis in $V$ as $\{\alpha_{i}\ |\ i\in S\}$ in one-to-one correspondence with S (with abuse of notation) and defining the symmetric bilinear form induced by the Schl\"afli matrix as
\begin{flalign}
 &B(\alpha_{i},\alpha_{j})=-\cos\left(\frac{\pi}{m_{ij}}\right),
\end{flalign}
the action of the $r_{i}$ on V could be realized as a reflection with respect to the hyperplane orthogonal to $\alpha_{i}$, $H_{\alpha_{i}}$,
\begin{flalign}
 &r_{i}v=v-2B(v,\alpha_{i})\alpha_{i},\label{reflectionCoxeter1}
\end{flalign}
with $v\in V$, the restriction of B on $\rm{Span}\{\alpha_{i},\alpha_{j}\}$ being positive semidefinite and nondegenerate. The bilinear form B is preserved by the action of $r_{i}$; $B(r_{i}v_{1},r_{i}v_{2})=B(v_{1},v_{2})$ for any $i\in S$ and $v_{1},v_{2}\in V$. If $\theta$ is the angle between $\alpha_{i}$ and $\alpha_{j}$ the action of $r_{i}r_{j}$ could be seen as a rotation of $2\theta$, \cref{rotation1}, in the plane spanned by $\alpha_{i}$ and $\alpha_{j}$. By the definition above, if $m_{ij}<\infty$ one recognizes $\theta$ to be $\pi/m_{ij}$. In this picture the meaning of $m_{ij}$, as order of the element $r_{i}r_{j}$ is evident. On the other hand if $m_{ij}=\infty$, taking $v=a\alpha_{i}+b\alpha_{j}$ we get
\begin{flalign}
 &r_{i}r_{j}v=v+2(a-b)(\alpha_{i}+\alpha_{j}), 
\end{flalign}
thus, acting iteratively (since $\alpha_{i}+\alpha_{j}$ is fixed by $r_{i}r_{j}$), we obtain
\begin{flalign}
 &(r_{i}r_{j})^{k}v=v+2k(a-b)(\alpha_{i}+\alpha_{j}), 
\end{flalign}
with $k\in\mathds{Z}$, that implies $r_{i}r_{j}$ has infinite order. From the geometric interpretation of the fundamental reflections it appears natural to define a correspondence between sets of vector in $V$ and Coxeter group. In particular we define a \textbf{root system} $\Phi$ in $V$ as a finite set of non-zero vectors in $V$ such that
\begin{enumerate}[label=(\roman*)]
 \item $\Phi\cap \mathds{R}\alpha=\{\alpha,\ -\alpha\}$ $\forall\ \alpha\in\Phi$
 \item $r_{\alpha}\Phi=\Phi$ $\forall\ \alpha\in\Phi$.
\end{enumerate}
The Coxeter group $W$ associated with $\Phi$ is the Coxeter group generated by all the reflections $s_{\alpha}$ with $\alpha\in\Phi$. A further refinement could be gained by defining a \textbf{simple system} $\Delta$ for $\Phi$ as a subset $\{\alpha_{i}\}$ of $\Phi$ such that  
\begin{enumerate}[label=(\roman*)]
 \item $\Delta$ is a basis for $\Phi\subseteq V$
 \item Any $\alpha\in \Phi$ could expressed as $\alpha=\sum_{i}a_{i}\alpha_{i}$ with all non-positive or non-negative coefficients $a_{i}$.
\end{enumerate}

Taking $W$ Coxeter group acting on $V$ with associated root system $\Phi$, if $\Delta$ is a simple system in $\Phi$ then $W$ is generated by the simple reflections $r_{\alpha_{i}}$  (we use also the notation $r_{i}$ in the next for a simple reflection corresponding to $\alpha_i$)  with $\alpha_{i}\in \Delta$.\\

Note that every reflection in the Coxeter group, $r_{\alpha}$, corresponds to a root $\alpha\in \Phi$, but not all elements of the Coxeter group (or Weyl group, as we will see) are in general reflections. The product of two reflections is not in general a reflection. This  explains why in general there are more elements in the Coxeter group than positive roots in the corresponding root system .\\

Any element of a Coxeter group could be expressed as words of simple reflections,
\begin{flalign}
 &r_{i_{1}}r_{i_2}...r_{i_{N}}.
\end{flalign}
Two words are equivalent if one could be obtained from the other by applying the founding relations \cref{Coxeterrelations1}. For example the sequences $r_{1}r_{3}r_{2}r_{2}$ and $r_{1}r_{3}$ are trivially equivalent; the same applies to $r_{2}r_{1}r_{3}r_{1}$ and $r_{2}r_{3}$ if $m_{31}=2$, while, if $m_{13}=3$, the former is equivalent to $r_{2}r_{3}r_{1}r_{3}$. The \textbf{length} $l(w)$ of a element $w$ in the Coxeter group $W$ is the smallest number of simple reflections $w$ could be written as product of.
The shortest expression of an element in a Coxeter group as product of simple reflections is called \textbf{reduced form} \cite{humphreys1990}. An element in the Coxeter group obtained as products of all simple reflection is called {\bf Coxeter element} and it could be shown that the Coxeter elements are all conjugate and have the same order. The order of the Coxeter elements is the {\bf Coxeter number} and it corresponds to the number of root divided by the rank.

\subsection{Weyl Group}
Weyl groups are particular cases of Coxeter groups and they play a fundamental role in the analysis next to come. Now we introduce Weyl groups and we describe their relation with Coxeter groups. Let's $\mathfrak{g}$ be a Lie algebra with Cartan matrix $A$, we denote with  $\mathfrak{h}$ its Cartan subalgebra, with $\Phi$ the set of its roots and with $\Delta$ the set of  simple roots. We define the \textbf{Weyl group} $W_{\mathfrak{g}}(A)$ of $\mathfrak{g}$ as the group generated by all the reflection $s_{\alpha}$, $\alpha\in \Phi$. Analogously to the Coxeter group case $W$ is generated by simple reflections $s_{\alpha}$, $\alpha\in \Delta$.\\
%

The $s_{\alpha}$ are reflections with respect to the hyperplanes orthogonal to the roots, called also \textbf{walls}, and their action on a weight $\Lambda\in \mathfrak{h}^{*}$ reads
\begin{flalign}
 &s_{\alpha}\Lambda=\Lambda-2\frac{\langle\alpha,\ \Lambda\rangle}{\langle\alpha,\ \alpha\rangle}\alpha,
\end{flalign}
where $\langle\ ,\ \rangle$ is the scalar product on the root system induced by the Killing form. The $s_{\alpha}$ preserve the scalar product,
\begin{flalign}
 &\scp{s_{\alpha}\Lambda}{s_{\alpha}\Sigma}=\scp{\Lambda}{\Sigma}.
\end{flalign}
This construction corresponds to a particular case of \cref{reflectionCoxeter1}.\\
A subgroup $G\subseteq GL(V)$ is said to be \textbf{cristallographic} if it stabilizes a lattice, $L\subseteq V$, i.e. $gL\subseteq L$ for all $g\in G$. The Weyl group of a Lie algebra is a cristallographic Coxeter group, leaving invariant the lattice of roots $\sum_{i}\mathbb{Z}\alpha_{i}$ where $i$ runs on simple roots.  The cristallographic property translates into the following additional requirement for the root system: 
\begin{flalign}
 &\frac{2\scp{\alpha}{\beta}}{\scp{\beta}{\beta}}\in \mathds{Z},
\end{flalign}
for any $\alpha,\ \beta\in \Phi$.
Any Weyl group is a cristallographic Coxeter group and the cristallographic property implies the Coxeter matrix entries $m_{ij}$, for $i\neq j$, could take only values in the set $\{2,3,4,6\}$\cite{humphreys1990}.\\

The Schl\"afli matrix is related to the Cartan matrix of the algebra (see section 5.3 in \cite{humphreys1990} for further details). Any generalized symmetrizable Cartan matrix $A$ could be written as product of a diagonal matrix $D$ with positive entries and a symmetric matrix $S$
\begin{flalign}
 &A=DS.
\end{flalign}
A possible choice is
\begin{subequations}
 \begin{flalign}
 &D_{ii}=\frac{1}{\scp{\alpha_{i}}{\alpha_{i}}}\\
 &S_{ij}=2\scp{\alpha_{i}}{\alpha_{j}}.
\end{flalign}
\end{subequations}
The relation between the Cartan matrix and the Schl\"afli matrix is explicitly given by
\begin{flalign}
 &C_{ij}=2\dfrac{\scp{\alpha_{i}}{\alpha_{j}}}{\sqrt{\scp{\alpha_{i}}{\alpha_{i}}  \scp{\alpha_{j}}{\alpha_{j}}}}=A_{ij}\sqrt{\dfrac{\scp{\alpha_{j}}{\alpha_{j}}}{\scp{\alpha_{i}}{\alpha_{i}}}},\label{eq:relationSchlafli1}
\end{flalign}
with the angle between two roots corresponding to the argument of cosine in \cref{schlafli1}. \Cref{eq:relationSchlafli1} implies
\begin{flalign}
 &C=\sqrt{D}A\sqrt{D}^{-1}.
\end{flalign}
We list in \cref{tab:mvalues} the possible angles between two simple roots and the corresponding $m$, the order of the product of their simple reflections, for the different connections appearing in the respective Dynkin diagram. The action of two consecutive reflections with respect to two planes orthogonal to a pair of vectors at angle $\pi/m$ corresponds to a rotation of $2\pi/m$ in the plane spanned by them.
\begin{table}[h!]
\renewcommand{\arraystretch}{2}
\begin{center}
\begin{tabular}{|c|c|c|c|}
\hline
\textbf{Dynkin diagram }&$\mathbf{\scp{\alpha}{\beta}}$&$\mathbf{\theta}$&$\mathbf{m_{\alpha\beta}}$\\ \hline
\scalebox{1} 
{
\begin{pspicture}(0,-0.2)(1.6,0.2)
\definecolor{color172b}{rgb}{0.996078431372549,0.996078431372549,0.996078431372549}
\pscircle[linewidth=0.02,dimen=outer,fillstyle=solid,fillcolor=color172b](0.2,0.0){0.2}
\pscircle[linewidth=0.02,dimen=outer,fillstyle=solid,fillcolor=color172b](1.4,0.0){0.2}
\end{pspicture} 
}
&0&$\pi/2$&2\\
\scalebox{1} 
{
\begin{pspicture}(0,-0.2)(1.6,0.2)
\definecolor{color172b}{rgb}{0.996078431372549,0.996078431372549,0.996078431372549}
\psline[linewidth=0.04cm](0.2,0.0)(1.4,0.0)
\pscircle[linewidth=0.02,dimen=outer,fillstyle=solid,fillcolor=color172b](0.2,0.0){0.2}
\pscircle[linewidth=0.02,dimen=outer,fillstyle=solid,fillcolor=color172b](1.4,0.0){0.2}
\end{pspicture} 
}
&-1&$\pi/3$&3\\
\scalebox{1} 
{
\begin{pspicture}(0,-0.22)(1.6,0.22)
\definecolor{color172b}{rgb}{0.996078431372549,0.996078431372549,0.996078431372549}
\psline[linewidth=0.04cm](0.6,0.0)(0.8,0.2)
\psline[linewidth=0.04cm](0.2,-0.1)(1.4,-0.1)
\psline[linewidth=0.04cm](0.2,0.1)(1.4,0.1)
\pscircle[linewidth=0.02,dimen=outer,fillstyle=solid,fillcolor=color172b](1.4,0.0){0.2}
\pscircle[linewidth=0.02,dimen=outer,fillstyle=solid,fillcolor=color172b](0.2,0.0){0.2}
\psline[linewidth=0.04cm](0.6,0.0)(0.8,-0.2)
\end{pspicture} 
}
&-1&$\pi/4$&4\\
\scalebox{1} 
{
\begin{pspicture}(0,-0.22)(1.6,0.22)
\definecolor{color172b}{rgb}{0.996078431372549,0.996078431372549,0.996078431372549}
\psline[linewidth=0.04cm](0.2,0.0)(1.4,0.0)
\psline[linewidth=0.04cm](0.6,0.0)(0.8,0.2)
\psline[linewidth=0.04cm](0.2,-0.1)(1.4,-0.1)
\psline[linewidth=0.04cm](0.2,0.1)(1.4,0.1)
\pscircle[linewidth=0.02,dimen=outer,fillstyle=solid,fillcolor=color172b](1.4,0.0){0.2}
\pscircle[linewidth=0.02,dimen=outer,fillstyle=solid,fillcolor=color172b](0.2,0.0){0.2}
\psline[linewidth=0.04cm](0.6,0.0)(0.8,-0.2)
\end{pspicture} 
}
&$-3/2$&$\pi/6$&6\\

\hline
\end{tabular}
\end{center}
\caption{For each type of joint between two simple roots in the Dynkin diagram, in the first column, we list the value of their scalar product, the angle between them and the value of $m$.}\label{tab:mvalues}
\end{table}
The signature of the generalized Cartan matrix is equal to the signature of the Schl\"afli matrix and a classification in finite, affine and indefinite types, identical to the one defined above, applies. 
\FloatBarrier
\section{Branes in $\mathbf{E_{11}}$}\label{sec:braneE11Weyl}
\FloatBarrier
In the previous section we have introduced some basic notions concerning Coxeter and Weyl groups. In this section we analyze the Kac-Mood algebra $E_{11}$ and the U-duality representations hosting half-supersymmetric branes in maximal supergravity theories. $E_{11}$ (or $E_{8}^{+++}$) is the Kac-Moody algebra obtained as very extension of $E_{8}$ \cite{kac1990}; its Dynkin diagram and Coxeter graph are sketched in \cref{E11CoxeterDynkindiagram}. From now on for the simple roots we adopt the numeration in \cref{E11Dynkindiagram}. Since $\det A=-2$, $E_{11}$ is of  indefinite type.\\   

The set of roots of a Kac-Moody algebra could be divided into real and imaginary roots, 
\begin{flalign}
 &\Phi=\Phi^{re}\sqcup\Phi^{im}\qquad \text{\rm{(disjoint union)}}.
\end{flalign}
A root $\alpha\in \Phi$ is called a \textbf{real root} if there exists $w\in W$ such that $\alpha=w\alpha_{i}$ for some $\alpha_{i}\in \Delta$, with $\Delta$ set of simple roots and the Weyl group $W$ being defined as the group generated by all simple reflections. A root that is not real is called \textbf{imaginary root}. Real and imaginary roots are completely characterized by their squared norm \cite{kac1990}:
\begin{subequations}
 \begin{align}
  \alpha\in \Phi^{re}\quad&\Longleftrightarrow\quad \scp{\alpha}{\alpha}>0\\
  \alpha\in \Phi^{im}\quad&\Longleftrightarrow\quad \scp{\alpha}{\alpha}\leqslant 0.
 \end{align}
\end{subequations}
%
%
%
%
%
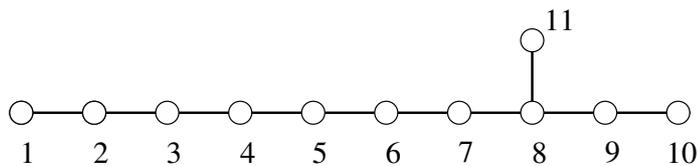
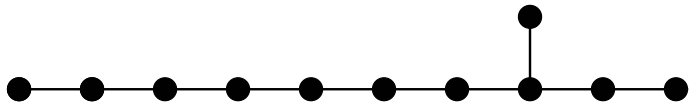
\begin{figure}[]
\centering
\subfloat[][Dynkin diagram of $E_{11}$]{
        \centering
\scalebox{0.8} 
{
\begin{pspicture}(0,-1.3184375)(11.413438,1.3184375)
\definecolor{color127b}{rgb}{0.996078431372549,0.996078431372549,0.996078431372549}
\psline[linewidth=0.04cm](8.6,0.7284375)(8.6,-0.4715625)
\psline[linewidth=0.04cm](0.2,-0.4715625)(11.2,-0.4715625)
\pscircle[linewidth=0.02,dimen=outer,fillstyle=solid](0.2,-0.4715625){0.2}
\pscircle[linewidth=0.02,dimen=outer,fillstyle=solid](1.4,-0.4715625){0.2}
\pscircle[linewidth=0.02,dimen=outer,fillstyle=solid,fillcolor=color127b](2.6,-0.4715625){0.2}
\pscircle[linewidth=0.02,dimen=outer,fillstyle=solid](3.8,-0.4715625){0.2}
\pscircle[linewidth=0.02,dimen=outer,fillstyle=solid](5.0,-0.4715625){0.2}
\pscircle[linewidth=0.02,dimen=outer,fillstyle=solid](9.8,-0.4715625){0.2}
\pscircle[linewidth=0.02,dimen=outer,fillstyle=solid](8.6,-0.4715625){0.2}
\pscircle[linewidth=0.02,dimen=outer,fillstyle=solid](7.4,-0.4715625){0.2}
\pscircle[linewidth=0.02,dimen=outer,fillstyle=solid](6.2,-0.4715625){0.2}
\pscircle[linewidth=0.02,dimen=outer,fillstyle=solid](11.0,-0.4715625){0.2}
\pscircle[linewidth=0.02,dimen=outer,fillstyle=solid](8.6,0.7284375){0.2}
\pscircle[linewidth=0.02,dimen=outer,fillstyle=solid,fillcolor=color127b](1.4,-0.4715625){0.2}
\pscircle[linewidth=0.02,dimen=outer,fillstyle=solid,fillcolor=color127b](0.2,-0.4715625){0.2}
\usefont{T1}{ptm}{m}{n}
\rput(0.30421874,-1.1265625){\Large 1}
\usefont{T1}{ptm}{m}{n}
\rput(1.5110937,-1.1265625){\Large 2}
\usefont{T1}{ptm}{m}{n}
\rput(2.7151563,-1.1265625){\Large 3}
\usefont{T1}{ptm}{m}{n}
\rput(3.9275,-1.1265625){\Large 4}
\usefont{T1}{ptm}{m}{n}
\rput(5.113125,-1.1265625){\Large 5}
\usefont{T1}{ptm}{m}{n}
\rput(6.3209376,-1.1265625){\Large 6}
\usefont{T1}{ptm}{m}{n}
\rput(7.5125,-1.1265625){\Large 7}
\usefont{T1}{ptm}{m}{n}
\rput(9.920625,-1.1265625){\Large 9}
\usefont{T1}{ptm}{m}{n}
\rput(8.72,-1.1265625){\Large 8}
\usefont{T1}{ptm}{m}{n}
\rput(11.060312,-1.1265625){\Large 10}
\usefont{T1}{ptm}{m}{n}
\rput(9.054218,1.0734375){\Large 11}
\end{pspicture} 
}
\label{E11Dynkindiagram}        
    }\\
 \subfloat[][Coxeter graph of $E_{11}$]{
        \centering
\scalebox{0.8} 
{
\begin{pspicture}(0,-0.8)(11.22,0.8)
\definecolor{color44b}{rgb}{0.00392156862745098,0.00392156862745098,0.00392156862745098}
\psline[linewidth=0.04cm](8.6,0.6)(8.6,-0.6)
\psline[linewidth=0.04cm](0.2,-0.6)(11.2,-0.6)
\pscircle[linewidth=0.02,dimen=outer,fillstyle=solid](0.2,-0.6){0.2}
\pscircle[linewidth=0.02,dimen=outer,fillstyle=solid](1.4,-0.6){0.2}
\pscircle[linewidth=0.02,dimen=outer,fillstyle=solid,fillcolor=color44b](2.6,-0.6){0.2}
\pscircle[linewidth=0.02,dimen=outer,fillstyle=solid,fillcolor=color44b](3.8,-0.6){0.2}
\pscircle[linewidth=0.02,dimen=outer,fillstyle=solid,fillcolor=color44b](5.0,-0.6){0.2}
\pscircle[linewidth=0.02,dimen=outer,fillstyle=solid,fillcolor=color44b](9.8,-0.6){0.2}
\pscircle[linewidth=0.02,dimen=outer,fillstyle=solid,fillcolor=color44b](8.6,-0.6){0.2}
\pscircle[linewidth=0.02,dimen=outer,fillstyle=solid,fillcolor=color44b](7.4,-0.6){0.2}
\pscircle[linewidth=0.02,dimen=outer,fillstyle=solid,fillcolor=color44b](6.2,-0.6){0.2}
\pscircle[linewidth=0.02,dimen=outer,fillstyle=solid,fillcolor=color44b](11.0,-0.6){0.2}
\pscircle[linewidth=0.02,dimen=outer,fillstyle=solid,fillcolor=color44b](8.6,0.6){0.2}
\pscircle[linewidth=0.02,dimen=outer,fillstyle=solid,fillcolor=color44b](1.4,-0.6){0.2}
\pscircle[linewidth=0.02,dimen=outer,fillstyle=solid,fillcolor=color44b](0.2,-0.6){0.2}
\end{pspicture} 
}

\label{E11Coxetergraph}    }
\caption{$E_{11}$ Dynkin and Coxeter diagrams.}\label{E11CoxeterDynkindiagram}

\end{figure}
It follows by the definition that the set of positive real roots, $\Phi^{re}_{+}$ could be generated by Weyl reflections acting on simple roots
\begin{flalign}
 &\Phi_{+}^{re}=W\Delta.
\end{flalign}
This implies that, since for any $\alpha\in \Phi^{re}$ there is $\alpha_{i}\in \Delta$ such that $\scp{\alpha}{\alpha}=\scp{\alpha_{i}}{\alpha_{i}}$, there could be real roots at most of $\rm{rank}\ \mathfrak{g}$ different lengths. We call $\alpha$ a \textbf{long real root} if $\scp{\alpha}{\alpha}=\max_{i}\scp{\alpha_{i}}{\alpha_{i}}$, a \textbf{short real root} if $\scp{\alpha}{\alpha}=\min_{i}\scp{\alpha_{i}}{\alpha_{i}}$. In the simply laced case real roots have only one possible length. This means that in $E_{11}$, normalizing the squared norm of simple roots to two, real roots have squared norm two and any imaginary root has squared norm zero or negative. This is a particular case of a more general result. It has been proved that the number of disjoint orbits for real roots corresponds to the number of disconnected components of the Dynkin diagram obtained deleting non single connection \cite[theorem 5.1 and corollaries 5.2 and 5.3]{Carbone:2010tn}. \\

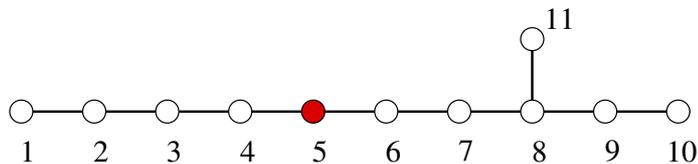
\begin{figure}[h!]
\centering
\scalebox{0.8} 
{
\begin{pspicture}(0,-1.3184375)(11.413438,1.3184375)
\definecolor{color127b}{rgb}{0.996078431372549,0.996078431372549,0.996078431372549}
\psline[linewidth=0.04cm](8.6,0.7284375)(8.6,-0.4715625)
\psline[linewidth=0.04cm](0.2,-0.4715625)(11.2,-0.4715625)
\pscircle[linewidth=0.02,dimen=outer,fillstyle=solid](0.2,-0.4715625){0.2}
\pscircle[linewidth=0.02,dimen=outer,fillstyle=solid](1.4,-0.4715625){0.2}
\pscircle[linewidth=0.02,dimen=outer,fillstyle=solid](2.6,-0.4715625){0.2}
\pscircle[linewidth=0.02,dimen=outer,fillstyle=solid](3.8,-0.4715625){0.2}
\pscircle[linewidth=0.02,dimen=outer,fillstyle=solid, fillcolor=red217](5.0,-0.4715625){0.2}
\pscircle[linewidth=0.02,dimen=outer,fillstyle=solid](9.8,-0.4715625){0.2}
\pscircle[linewidth=0.02,dimen=outer,fillstyle=solid](8.6,-0.4715625){0.2}
\pscircle[linewidth=0.02,dimen=outer,fillstyle=solid](7.4,-0.4715625){0.2}
\pscircle[linewidth=0.02,dimen=outer,fillstyle=solid](6.2,-0.4715625){0.2}
\pscircle[linewidth=0.02,dimen=outer,fillstyle=solid](11.0,-0.4715625){0.2}
\pscircle[linewidth=0.02,dimen=outer,fillstyle=solid](8.6,0.7284375){0.2}
\usefont{T1}{ptm}{m}{n}
\rput(0.30421874,-1.1265625){\Large 1}
\usefont{T1}{ptm}{m}{n}
\rput(1.5110937,-1.1265625){\Large 2}
\usefont{T1}{ptm}{m}{n}
\rput(2.7151563,-1.1265625){\Large 3}
\usefont{T1}{ptm}{m}{n}
\rput(3.9275,-1.1265625){\Large 4}
\usefont{T1}{ptm}{m}{n}
\rput(5.113125,-1.1265625){\Large 5}
\usefont{T1}{ptm}{m}{n}
\rput(6.3209376,-1.1265625){\Large 6}
\usefont{T1}{ptm}{m}{n}
\rput(7.5125,-1.1265625){\Large 7}
\usefont{T1}{ptm}{m}{n}
\rput(9.920625,-1.1265625){\Large 9}
\usefont{T1}{ptm}{m}{n}
\rput(8.72,-1.1265625){\Large 8}
\usefont{T1}{ptm}{m}{n}
\rput(11.060312,-1.1265625){\Large 10}
\usefont{T1}{ptm}{m}{n}
\rput(9.054218,1.0734375){\Large 11}
\end{pspicture} 
}
\caption{Decomposition of $E_{11}$ in $GL(5,\mathds{R})\times E_{6(6)}$.}\label{E11DynkindiagramD5}

\end{figure}
At this point we need to recall the role of $E_{11}$ in the context of maximal supergravity theories. Starting from the eleven dimensional $E_{11}$ non-linear realization of M-theory it is possible to derive the bosonic spectrum of all maximal supergravity theories from three dimensions above \cite{Riccioni:2007au}. This can be achieved by decomposing the adjoint representation of $E_{11}$ in the subgroup $E_{11-d}\times GL(d,\mathds{R})$ and selecting real roots. In particular for the $d$ dimensional maximal theory, in order to define the opportune decomposition, one should identify the gravity line, i.e. the subset of nodes of the $E_{11}$ Dynkin diagram containing node 1, following the numeration defined in \cref{E11CoxeterDynkindiagram}, and corresponding to the Dynkin diagram of $A_{d-1}$. The nodes joined by a single connection to the last node of $A_{d-1}$, and not belonging to it, should be deleted. The remaining nodes correspond to the Dynkin diagram of $E_{11-d}$, the U-duality symmetry of the $d$ dimensional maximal supergravity. The $SL(d,\mathds{R})$ symmetry described by the gravity line is promoted to a $GL(d,\mathds{R})$ by one extra Cartan generator coming from the deleted nodes. This symmetry describes the gravity sector of the theory.  The procedure could be visualized, for the five dimensional theory, in \cref{E11DynkindiagramD5}. Node 5 is deleted, the first four nodes plus the Cartan associated with $\alpha_{5}$ define a symmetry $GL(5,\mathds{R})$, while the nodes associated with $\alpha_{i}$ for $i>5$ correspond to the Dynkin diagram of the five dimensional U-duality group $E_{6(6)}$. Carrying out this procedure, performing the branching and selecting states corresponding to real roots, one obtains for every maximal supergravity theory the spectrum of differential forms. We limit our attention to maximal theories from three to nine dimensions. We list the results in \cref{weightsd3,weightsd4,weightsd5,weightsd6,weightsd7,weightsd8,weightsd9}, where, in the first column we define the label for the highest weight of the representation. We use a notation of the form $\Lambda_{dp}$ where $d$ is the dimension and $p$ the rank of the corresponding differential form. If more than one representation occurs for the same rank forms these ere distinguished with $a,b$ after $p$ in the subscript. In the second column of the tables we show the coordinates $a_{i}$ in the basis of simple roots, $\Lambda=\sum_{i=1}^{11}a_{i}\alpha_{i}$.  In the third column we list the dimension of the irreducible representations of the corresponding U-duality group and in the last column the corresponding Dynkin labels. All these representations correspond in $E_{11}$ to real roots \cite{Riccioni:2007au}. This means that for any two weights belonging to any of these U-duality irreps. there is a Weyl transformation in $W_{E_{11}}$ connecting them. \\

\begin{table}[h!]
\renewcommand{\arraystretch}{2}
\begin{center}
\begin{tabular}{|c|c|c|c|}
\hline
\textbf{Weight}&\textbf{Vector}&\textbf{$\mathbf{GL(2,\mathds{R})}$ rep}&\textbf{Dynkin labels} \\
\hline\hline
 $\Lambda_{99}$& (1, 2, 3, 4, 5, 6, 7, 8, 5, 4, 4)&\bf{4}&$\boxed{3}$\\
$\Lambda_{98}$& (1, 2, 3, 4, 5, 6, 7, 8, 4, 3,  4)&\bf{3}&$\boxed{2}$\\
$\Lambda_{97}$&(1, 2, 3, 4, 5, 6, 7, 7, 4, 3, 3)&\bf{3}&$\boxed{2}$\\
$\Lambda_{96b}$&(1, 2, 3, 4, 5, 6, 6, 6, 3, 2, 3)&\bf{2}&$\boxed{1}$\\
$\Lambda_{96a}$& (1, 2, 3, 4, 5, 6, 6, 6, 4, 2, 2)&\bf{1}&$\boxed{0}$\\
$\Lambda_{95}$& (1, 2, 3, 4, 5, 5, 5, 5, 3, 2, 2)&\bf{2}&$\boxed{1}$\\
$\Lambda_{94}$& (1, 2, 3, 4, 4, 4, 4, 4, 2, 1, 2)&\bf{1}&$\boxed{0}$\\
$\Lambda_{93}$& (1, 2, 3, 3, 3, 3, 3, 3, 2, 1, 1)&\bf{1}&$\boxed{0}$\\
$\Lambda_{92}$& (1, 2, 2, 2, 2, 2, 2, 2, 1, 1, 1)&\bf{2}&$\boxed{1}$\\
$\Lambda_{91b}$& (1, 1, 1, 1, 1, 1, 1, 1, 0, 0, 1)&\bf{1}&$\boxed{0}$\\
$\Lambda_{91a}$& (1, 1, 1, 1, 1, 1, 1, 1, 1, 1, 0)&\bf{2}&$\boxed{1}$\\
\hline

\end{tabular}
\end{center}
\caption{Highest weights in $E_{11}$ for the U-duality irreps. hosting the differential forms in nine dimensional maximal supergravity.}\label{weightsd9}
\end{table}

\begin{table}[h!]
\renewcommand{\arraystretch}{2}
\begin{center}
\begin{tabular}{|c|c|c|c|}
\hline
\textbf{Weight}&\textbf{Vector}&\textbf{$\mathbf{SL(3,\mathds{R})\times SL(2,\mathds{R})}$ rep}& \bf{Dynkin Labels}\\
\hline\hline
$\Lambda_{88}$& (1, 2, 3, 4, 5, 6, 7, 8, 7, 4, 4)&\bf{15}&$\boxed{2\ 1\ 0}$\\
$\Lambda_{87}$& (1, 2, 3, 4, 5, 6, 7, 7, 6, 3, 4)&\bf{12}&$\boxed{2\ 0\ 1}$\\
$\Lambda_{86b}$& (1, 2, 3, 4, 5, 6, 6, 6, 5, 3, 3)&\bf{8}&$\boxed{1\ 1\ 0}$\\
$\Lambda_{86a}$& (1, 2, 3, 4, 5, 6, 6, 6, 4, 2, 4)&\bf{3}&$\boxed{0\ 0\ 2}$\\
$\Lambda_{85}$& (1, 2, 3, 4, 5, 5, 5, 5, 4, 2, 3)&\bf{6}&$\boxed{1\ 0\ 1}$\\
$\Lambda_{84}$& (1, 2, 3, 4, 4, 4, 4, 4, 3, 2, 2)&\bf{3}&$\boxed{0\ 1\ 0}$\\
$\Lambda_{83}$& (1, 2, 3, 3, 3, 3, 3, 3, 2, 1, 2)&\bf{2}&$\boxed{0\ 0\ 1}$\\
$\Lambda_{82}$& (1, 2, 2, 2, 2, 2, 2, 2, 2, 1, 1)&\bf{3}&$\boxed{1\ 0\ 0}$\\
$\Lambda_{81}$& (1, 1, 1, 1, 1, 1, 1, 1, 1, 1, 1)&\bf{6}&$\boxed{0\ 1\ 1}$\\
\hline

\end{tabular}
\end{center}
\caption{Highest weights in $E_{11}$ for the U-duality irreps. hosting the differential forms in eight dimensional maximal supergravity. }\label{weightsd8}
\end{table}

\begin{table}[h!]
\renewcommand{\arraystretch}{1.6}
\begin{center}
\begin{tabular}{|c|c|c|c|}
\hline
\textbf{Weight}&\textbf{Vector}& \textbf{$\mathbf{SL(5,\mathds{R})}$ rep}&\bf{Dynkin Labels} \\
\hline\hline
$\Lambda_{77}$& (1, 2, 3, 4, 5, 6, 7, 10, 7, 4, 6)&\bf{70}&$\boxed{0\ 0\ 1\ 2}$\\
$\Lambda_{76b}$& (1, 2, 3, 4, 5, 6, 6, 8, 6, 4, 4)&\bf{15}&$\boxed{0\ 0\ 2\ 0}$\\
$\Lambda_{76a}$& (1, 2, 3, 4, 5, 6, 6, 9, 6, 3, 5)&\bf{40}&$\boxed{1\ 0\ 0\ 1}$\\
$\Lambda_{75}$& (1, 2, 3, 4, 5, 5, 5, 7, 5, 3, 4)&\bf{24}&$\boxed{0\ 0\ 1\ 1}$\\
$\Lambda_{74}$& (1, 2, 3, 4, 4, 4, 4, 6, 4, 2, 3)&\bf{10}&$\boxed{1\ 0\ 0\ 0}$\\
$\Lambda_{73}$& (1, 2, 3, 3, 3, 3, 3, 4, 3, 2, 2)&\bf{5}&$\boxed{0\ 0\ 1\ 0}$\\
$\Lambda_{72}$& (1, 2, 2, 2, 2, 2, 2, 3, 2, 1, 2)&\bf{5}&$\boxed{0\ 0\ 0\ 1}$\\
$\Lambda_{71}$& (1, 1, 1, 1, 1, 1, 1, 2, 2, 1, 1)&\bf{10}&$\boxed{0\ 1\ 0\ 0}$\\
\hline

\end{tabular}
\end{center}
\caption{Highest weights in $E_{11}$ for the U-duality irreps. hosting the differential forms in seven dimensional maximal supergravity.}\label{weightsd7}
\end{table}

\begin{table}[h!]
\renewcommand{\arraystretch}{1.6}
\begin{center}
\begin{tabular}{|c|c|c|c|}
\hline
\textbf{Weight}&\textbf{Vector}&\textbf{$\mathbf{SO(5,5)}$ rep}&\bf{Dynkin Labels} \\
\hline\hline
$\Lambda_{66a}$& (1, 2, 3, 4, 5, 6, 9, 12, 9, 5, 6)&\bf{320}&$\boxed{0\ 0\ 1\ 1\ 0}$\\
$\Lambda_{66b}$& (1, 2, 3, 4, 5, 6, 10, 12, 8, 4, 6)&\bf{126}&$\boxed{2\ 0\ 0\ 0\ 0}$\\
$\Lambda_{65}$& (1, 2, 3, 4, 5, 5, 8, 10, 7, 4, 5)&\bf{144}&$\boxed{1\ 0\ 0\ 1\ 0}$\\
$\Lambda_{64}$& (1, 2, 3, 4, 4, 4, 6, 8, 6, 3, 4)&\bf{45}&$\boxed{0\ 0\ 1\ 0\ 0}$\\
$\Lambda_{63}$& (1, 2, 3, 3, 3, 3, 5, 6, 4, 2, 3)&\bf{16}&$\boxed{1\ 0\ 0\ 0\ 0}$\\
$\Lambda_{62}$& (1, 2, 2, 2, 2, 2, 3, 4, 3, 2, 2)&\bf{10}&$\boxed{0\ 0\ 0\ 1\ 0}$\\
$\Lambda_{61}$& (1, 1, 1, 1, 1, 1, 2, 3, 2, 1, 2)&\bf{16}&$\boxed{0\ 0\ 0\ 0\ 1}$\\
\hline

\end{tabular}
\end{center}
\caption{Highest weights in $E_{11}$ for the U-duality irreps. hosting the differential forms in six dimensional maximal supergravity.}\label{weightsd6}
\end{table}

\begin{table}[h!]
\renewcommand{\arraystretch}{1.6}
\begin{center}
\begin{tabular}{|c|c|c|c|}
\hline
\textbf{Weight}&\textbf{Vector}&\textbf{$\mathbf{E_{6(6)}}$ rep}&\bf{Dynkin Labels} \\
\hline\hline
$\Lambda_{55}$& (1, 2, 3, 4, 5, 9, 12, 15, 10, 5, 8)&\bf{1728}&$\boxed{1\ 0\ 0\ 0\ 0\ 1}$\\
$\Lambda_{54}$& (1, 2, 3, 4, 4, 7, 10, 12, 8, 4, 6)&\bf{351}&$\boxed{0\ 1\ 0\ 0\ 0\ 0}$\\
$\Lambda_{53}$& (1, 2, 3, 3, 3, 5, 7, 9, 6, 3, 5)&\bf{78}&$\boxed{0\ 0\ 0\ 0\ 0\ 1}$\\
$\Lambda_{52}$& (1, 2, 2, 2, 2, 4, 5, 6, 4, 2, 3)&\bf{27}&$\boxed{1\ 0\ 0\ 0\ 0\ 0}$\\
$\Lambda_{51}$& (1, 1, 1, 1, 1, 2, 3, 4, 3, 2, 2)&\bf{27}&$\boxed{0\ 0\ 0\ 0\ 1\ 0}$\\
\hline
\end{tabular}
\end{center}
\caption{Highest weights in $E_{11}$ for the U-duality irreps. hosting the differential forms in five dimensional maximal supergravity.}\label{weightsd5}
\end{table}

\begin{table}[h!]
\renewcommand{\arraystretch}{1.6}
\begin{center}
\begin{tabular}{|c|c|c|c|}
\hline
\textbf{Weight}&\textbf{Vector}&$\mathbf{E_{7(7)}}$ \bf{rep}&\bf{Dynkin Labels} \\
\hline\hline
$\Lambda_{44}$& (1, 2, 3, 4, 8, 12, 16, 20, 14, 7, 10)&\bf{8645}&$\boxed{0\ 0\ 0\ 0\ 1\ 0\ 0}$\\
$\Lambda_{43}$& (1, 2, 3, 3, 6, 9, 12, 15, 10, 5, 8)&\bf{912}&$\boxed{0\ 0\ 0\ 0\ 0\ 0\ 1}$\\
$\Lambda_{42}$& (1, 2, 2, 2, 4, 6, 8, 10, 7, 4, 5)&\bf{133}&$\boxed{0\ 0\ 0\ 0\ 0\ 1\ 0}$\\
$\Lambda_{41}$& (1, 1, 1, 1, 3, 4, 5, 6, 4, 2, 3)&\bf{56}&$\boxed{1\ 0\ 0\ 0\ 0\ 0\ 0}$\\
\hline

\end{tabular}
\end{center}
\caption{Highest weights in $E_{11}$ for the U-duality irreps. hosting the differential forms in four dimensional maximal supergravity.}\label{weightsd4}
\end{table}

\begin{table}[h!]
\renewcommand{\arraystretch}{2}
\begin{center}
\begin{tabular}{|c|c|c|c|}
\hline
\textbf{Weight}&\textbf{Vector}& \textbf{$\mathbf{E_{8(8)}}$ rep}&\bf{Dynkin Labels} \\
\hline\hline
$\Lambda_{33}$& (1, 2, 3, 9, 15, 21, 27, 33, 22, 11, 17)&\bf{147250}&$\boxed{0\ 0\ 0\ 0\ 0\ 0\ 0\ 1}$\\
$\Lambda_{32}$& (1, 2, 2, 6, 10, 14, 18, 22, 15, 8, 11)&\bf{3875}&$\boxed{0\ 0\ 0\ 0\ 0\ 0\ 1\ 0}$\\
$\Lambda_{31}$& (1, 1, 1, 4, 6, 8, 10, 12, 8, 4, 6)&\bf{248}&$\boxed{1\ 0\ 0\ 0\ 0\ 0\ 0\ 0}$\\
\hline
\end{tabular}
\end{center}
\caption{Highest weights in $E_{11}$ for the U-duality irreps. hosting the differential forms in three dimensional maximal supergravity.}\label{weightsd3}
\end{table}

Looking in perspective to brane solutions in maximal supergravity the analysis of differential forms is crucial since a $p$-brane is charged with respect to $(p+1)$-forms. This link induces an algebraic characterization for the 1/2-BPS solutions \cite{Bergshoeff:2013sxa, Bergshoeff:2014lxa}. Half-supersymmetric branes are solutions preserving the maximum amount of supersymmetry and they serve also as building blocks for less supersymmetric solutions. It has been found \cite{Bergshoeff:2013sxa, Bergshoeff:2014lxa} that 1/2-BPS branes in maximal supergravity correspond to the longest weights of the U-duality representation hosting their charges. This correspondence, taking the name of \textit{longest weight rule}, defines and elegant criterion to identify the number of half-supersymmetric solutions in any maximal supergravity theory and it turns out to play a prominent role also in the classification of U-duality orbits \cite{Bergshoeff:2013sxa, Bergshoeff:2014lxa, Marrani:2015gfa}. It should be remarked that the length in this case is computed with respect to the U-duality algebra and thus it does not correspond in general to the length in $E_{11}$.\\

The results on the orbits of the real roots combined with the algebraic classification of 1/2-BPS branes describe a really interesting setting. Any pair of half-supersymmetric branes, say a $p_1$-brane and a $p_2$-brane, taken in any two maximal theories in $d_{1}$ and $d_{2}$ are connected by a Weyl reflection in $E_{11}$. In the next section we use the results just discussed as a starting point to investigate the relation between different branes in different theories and to define universal algebraic structures codifying the number of 1/2 p-brane in any dimension.

\FloatBarrier


\section{Algebraic Structures Behind Half-Supersymmetric Branes}\label{sec:3}

The notions introduced in the first section and the general analysis of the previous one evidence the relevant role that the Weyl groups play in the analysis of half-supersymmetric solutions in maximal theories. In this section, studying the action of the Weyl group of the U-duality groups on branes, we derive a set of algebraic relations that fully define the content of half-supersymmetric p-branes in any maximal supergravity from three to nine dimensions. As a first step in this direction, let's consider a Lie algebra $\mathfrak{g}$  with Weyl group $W_{\mathfrak{g}}$ acting on an irreducible representation of $\mathfrak{g}$, $V$. We call $\Lambda$ the highest weight of $V$ and 
\begin{flalign}
 &\Lambda=\boxed{d_{1}\ d_{2}\ ...\ d_{n}}
\end{flalign}
its Dynkin labels, characterized by $d_{i}\geqslant 0$ for all $i=1,..,n$. We want to identify its stabilizers inside $W_{\mathfrak{g}}$. To this aim it will be useful to introduce the concept of parabolic subgroup of a Coxeter system. 
Give a Coxeter system $(W,S)$ a \textbf{parabolic subgroup} for $W$, $W_{I}$ is a subgroup of $W$ generated by all the simple reflections in the subset $I\subseteq S$. This induces also the definition of its complement
\begin{flalign}
&W^{I}=\{w\in W\ |\ l(ws_{\alpha})>l(w)\ \forall\ s_{\alpha}\in I\}.
\end{flalign}
We refer to the set of simple root generating $W_{I}$ as $\Delta_{I}$. The isotropy group of the highest weight of the  representation $V$ could be seen as a parabolic subgroup of $W_{\mathfrak{g}}$.  In particular we consider the following application of \cite[proposition 1.15]{humphreys1990}
\begin{proposition}\label{prop:isitropygroup}
 Let $\Lambda$ be a dominant weight in an irreducible representation of the Lie algebra $\mathfrak{g}$ then its isotropy group in $W_{\mathfrak{g}}$ is the parabolic subgroup $W_{I^{\Lambda}_{0}}$, where $I^{\Lambda}_{0}=\{s_{\alpha_{i}}\in S\ |\ \scp{\Lambda}{\alpha_{i}}=0\}$
\end{proposition}
We report the proof for completeness, referring to \cite{humphreys1990} for the necessary results.
\begin{proof}
 Let's take $\Lambda$ dominant weight, then 
 \begin{flalign*}
  &\scp{\Lambda}{\alpha_{i}}\geqslant 0\qquad \forall\ \alpha_{i}\in\Delta.
 \end{flalign*}
It is clear that any $w\in W_{I^{\Lambda}_{0}}$ stabilizes $\Lambda$; now we want to show that any stabilizer belong to $W_{I^{\Lambda}_{0}}$. Assume there is $w\notin W_{I^{\Lambda}_{0}}$ such that $w\Lambda=\Lambda$. $w$ can be uniquely decomposed (by \cite[proposition 1.10]{humphreys1990}) as $w=uv$ with $u\in W^{I^{\Lambda}_{0}}$ and $v\in W_{I^{\Lambda}_{0}}$. Thus $w\Lambda=uv\Lambda=u\Lambda=\Lambda$. Then 
\begin{flalign*}
 &l(us_{\alpha})>l(u)\qquad \forall\ \alpha\in\Delta_{I}.
\end{flalign*}
This implies (by \cite[1.6 and corollary 1.7]{humphreys1990}) 
\begin{flalign*}
 &u\Delta_{I}\subset\Phi^{+}.
\end{flalign*}
There should be $\alpha_{i}\in\Delta$ such that $u\alpha_{i}<0$ and by the argument just exposed $\alpha_{i}\notin\Delta_{I}$. Thus we get
 \begin{flalign}
  &\scp{\Lambda}{\alpha_{i}}> 0,
 \end{flalign}
by definition of dominant weight, and
\begin{flalign}
 &\scp{\Lambda}{\alpha_{i}}=\scp{u\Lambda}{u\alpha_{i}}=\scp{\Lambda}{u\alpha_{i}}\leqslant 0
\end{flalign}
 that is absurd. 
\end{proof}

Then if $\Lambda_{1}$ and $\Lambda_{2}$ are two weights connected by a Weyl reflection $s$, $\Lambda_{2}=s\Lambda_{1}$ and $w$ is a stabilizer for $\Lambda_{1}$, $s^{-1}ws$ is a stabilizer for $\Lambda_{2}$. This defines a correspondence between stabilizers inside the Weyl group acting on Weyl equivalent weights. \\

Now we consider the following theorem , as a specialization of \cite[proposition 1.15 and theorem 1.12]{humphreys1990}.

\begin{thm} [Weyl Orbit]\label{thm:weylorbit1}
Given a weight $\Lambda$ in an irreducible representation of a Lie algebra $\mathfrak{g}$, its orbits under the Weyl group $W_{\mathfrak{g}}$ has dimension $N$ given by
\begin{flalign}
 &N=\dfrac{\dim W_{\mathfrak{g}}}{\dim W_{I^{\Lambda}_{0}}},
\end{flalign}
where $W_{I^{\Lambda}_{0}}$ is its isotropy group in $W_{\mathfrak{g}}$.
\end{thm}
\begin{proof}
Consider $W_{I^{\Lambda}_{0}}$, the isotropy group of $\Lambda$. Any $w\in W_{\mathfrak{g}}$ could be decomposed uniquely as
\begin{flalign*}
 w=uv
\end{flalign*}
with $u\in W^{I^{\Lambda}_{0}} $, $v\in W_{I^{\Lambda}_{0}}$ and
\begin{flalign*}
 &l(us_{\alpha})>l(u)\qquad \forall\alpha\in\Delta_{I}.
\end{flalign*}
This means the sets $uW_{I^{\Lambda}_{0}}$ for different $u\in W^{I^{\Lambda}_{0}}$ are disjoint. For any weight $\Lambda_i$ connected to $\Lambda$ by a Weyl transformation $u_{i}$, $\Lambda_{i}=u_{i}\Lambda$ we have
\begin{flalign*}
 &u_{i}\in W^{I^{\Lambda}_{0}}
\end{flalign*}
and $u_{i}$ is unique. By definition any element of $u_{i}W_{I^{\Lambda}_{0}}$ brings $\Lambda$ to $\Lambda_{i}$. These are exactly $\dim W_{I^{\Lambda_1}_{0}}$ elements. By applying the same arguments to all the weights in the Weyl orbit one gets the result.
\end{proof}
%

By \cref{prop:isitropygroup} and \cref{thm:weylorbit1} the dimension of the orbit of a dominant weight, under the action of the Weyl group, in an irreducible representation of a Lie algebra $\mathfrak{g}$ is the dimension of the Weyl group of $\mathfrak{g}$ divided by the dimension of the Weyl group associated with the subalgebra identified by its zero Dynkin labels, i.e its isotropy group in $W_{\mathfrak{g}}$.\\

By virtue of the \textit{longest weight rule} we could immediately apply \cref{thm:weylorbit1} to find the number of branes in maximal theories from three to nine dimensions. Looking at the Dynkin labels of the highest weights of the U-duality representations appearing in \cref{weightsd3,weightsd4,weightsd5,weightsd6,weightsd7,weightsd8,weightsd9} we realize that the number of half-supersymmetric  branes in $d$ dimensions, rank by rank, is given by the following relations
\begin{subequations}
\begin{flalign}
&N^d_{0\text{-}brane} =\dfrac{\dim W_{E_{11-d}}}{\dim W_{E_{10-d}}} \label{eq:0brane}\\
&N^d_{1\text{-}brane} =\dfrac{\dim W_{E_{11-d}}}{\dim W_{D_{10-d}}} \label{eq:1brane}\\
&N^d_{2\text{-}brane} =\dfrac{\dim W_{E_{11-d}}}{\dim W_{A_{10-d}}} \label{eq:2brane}\\
&N^d_{3\text{-}brane} =\dfrac{\dim W_{E_{11-d}}}{\dim W_{A_{1}\times A_{9-d}}} \label{eq:3brane}\\
&N^d_{4\text{-}brane} =\dfrac{\dim W_{E_{11-d}}}{\dim W_{A_{9-d}}} \label{eq:4brane}\\
&N^d_{5\text{-}brane} =\dfrac{\dim W_{E_{11-d}}}{\dim W_{A_{9-d}}}+\dfrac{\dim W_{E_{11-d}}}{\dim W_{A_{10-d}}} \label{eq:5brane}\\
&N^d_{6\text{-}brane} =\dfrac{\dim W_{E_{11-d}}}{\dim W_{A_{9-d}}} \label{eq:6brane}\\
&N^d_{7\text{-}brane} =\dfrac{\dim W_{E_{11-d}}}{\dim W_{A_{9-d}}} \label{eq:7brane}\\ 
&N^d_{8\text{-}brane} =\dfrac{\dim W_{E_{11-d}}}{\dim W_{A_{9-d}}}.\label{eq:8brane} 
\end{flalign}\label{eq:relation0}
\end{subequations}
It is remarkable that the relations just found describe the content of half-supersymmetric solutions, rank by rank in any dimensions, despite these are standard or non-standard branes. We have obtained five types of different relations; branes with rank lower than five are governed by five different rules, while, for solutions of rank four and higher, the same relation holds, with an additional contribution for 5-branes, induced by the fact that these couple both with vector and tensor multiplets, identical to \cref{eq:2brane}. We report the chain of embeddings of the Lie algebras the Weyl groups appearing in the denominator of \cref{eq:relation0} correspond to
\begin{flalign}
 &\left.\begin{tabular}{c}
   $E_{10-d}$\\
   $D_{10-d}$\\
   $A_{10-d}$
  \end{tabular}\right\}\supset A_{1}\times A_{9-d}\supset A_{9-d}.\label{eq:embeddings}
\end{flalign}
The isotropy groups appearing in \cref{eq:relation0} are Weyl groups of rank $10-d$ algebras for 0- to 3-branes and rank $9-d$ for 4-branes and higher rank solutions, with the exception described above for 5-branes. Moreover we note that the first relation, \cref{eq:0brane}, reproduces exactly the number of 0-branes also in the nine dimensional theory, where these belong to two different representations, identifying $E_{10-d}$ with the symmetry group of the two possible ten dimensional uplifts, type IIA and IIB theories.\\
It could also happen that different types of rules give the same number of branes. This is the case, for example, of the 0- and 1-branes in five dimensions, due to the fact that $E_{5}\sim D_{5}$. By the same way these relations make explicit that in six dimensions there is the same number of half-supersymmetric 0-branes and 2-branes. The same is true for 0-brane and 3-brane in seven dimensions, for 0-branes and 4-branes and 1-branes and 3-branes in eight dimensions.\\

In \cref{tab:branesmaximal1} we list the number of half-supersymmetric solutions in any maximal supergravity theory and the dimension of the Weyl group of the U-duality group.
\begin{table}[h!]
\renewcommand{\arraystretch}{2.5}
\begin{center}
\resizebox{\textwidth}{!}{
\begin{tabular}{|c|c|ccccccccc|}
\hline
\bf{d}&\bf{dim}$\mathbf{W_{E_{11-d}} }$&\textbf{1-f}&\textbf{2-f}&\textbf{3-f}&\textbf{4-f}&\textbf{5-f}&\textbf{6-f}&\textbf{7-f}&\textbf{8-f}&\textbf{9-f}\\ \hline \hline
9&2&1+2&2&1&1&2&\multicolumn{1}{c|}{1+2}&
\renewcommand{\arraystretch}{1}
\begin{tabular}{c}
 3\\
 \bl{2}
\end{tabular}
&
\renewcommand{\arraystretch}{1}
\begin{tabular}{c}
 3\\
 \bl{2}
\end{tabular}&
\renewcommand{\arraystretch}{1}
\begin{tabular}{c}
 4\\
 \bl{2}
\end{tabular}\\  \cline{8-8}

8&12&6&3&2&3&\multicolumn{1}{c|}{6}&\renewcommand{\arraystretch}{1}
\begin{tabular}{c}
 3+8\\
 \bl{2+6}
\end{tabular}&
\renewcommand{\arraystretch}{1}
\begin{tabular}{c}
 12\\
 \bl{6}
\end{tabular}&
\renewcommand{\arraystretch}{1}
\begin{tabular}{c}
 15\\
 \bl{6}
\end{tabular}&\\  \cline{7-7}

7&120&10&5&5&\multicolumn{1}{c|}{10}&\renewcommand{\arraystretch}{1}
\begin{tabular}{c}
 24\\
 \bl{20}
\end{tabular}&
\renewcommand{\arraystretch}{1}
\begin{tabular}{c}
 15+40\\
 \bl{5+20}
\end{tabular}&
\renewcommand{\arraystretch}{1}
\begin{tabular}{c}
 70\\
 \bl{20}
\end{tabular}&&\\  \cline{6-6}

6&1920&16&10&\multicolumn{1}{c|}{16}&
\renewcommand{\arraystretch}{1}
\begin{tabular}{c}
 45\\
 \bl{40}
\end{tabular}&
\renewcommand{\arraystretch}{1}
\begin{tabular}{c}
 144\\
 \bl{80}
\end{tabular}&
\renewcommand{\arraystretch}{1}
\begin{tabular}{c}
 126+320\\
 \bl{16+80}
\end{tabular}&&&\\  \cline{5-5}

5&51840&27&\multicolumn{1}{c|}{27}&
\renewcommand{\arraystretch}{1}
\begin{tabular}{c}
 78\\
 \bl{72}
\end{tabular}&
\renewcommand{\arraystretch}{1}
\begin{tabular}{c}
 351\\
 \bl{216}
\end{tabular}&
\renewcommand{\arraystretch}{1}
\begin{tabular}{c}
 1728\\
 \bl{432}
\end{tabular}&&&&\\  \cline{4-4}

4&2903040&\multicolumn{1}{c|}{56}&
\renewcommand{\arraystretch}{1}
\begin{tabular}{c}
 133\\
 \bl{126}
\end{tabular}&
\renewcommand{\arraystretch}{1}
\begin{tabular}{c}
 912\\
 \bl{576}
\end{tabular}&
\renewcommand{\arraystretch}{1}
\begin{tabular}{c}
 8645\\
 \bl{2016}
\end{tabular}&&&&&\\ \cline{3-3}

3&696729600&
\renewcommand{\arraystretch}{1}
\begin{tabular}{c}
 248\\
 \bl{240}
\end{tabular}&
\renewcommand{\arraystretch}{1}
\begin{tabular}{c}
 3875\\
 \bl{2160}
\end{tabular}&
\renewcommand{\arraystretch}{1}
\begin{tabular}{c}
 147250\\
 \bl{17280}
\end{tabular}&&&&&&\\

\hline

\end{tabular}
}
\end{center}
\caption{For any $d$ dimensional maximal theory we list the the dimension of the Weyl group of the U-duality group $E_{11-d}$, the dimension of the representations hosting differential forms and the number of components coupling to half-supersymmetric branes. p-{\textbf{f}} denotes the rank of the differential forms. When the number of half-supersymmetric solutions does not correspond to the dimension of the representation it appears in blue (non-standard branes \cite{Bergshoeff:2013sxa, Bergshoeff:2014lxa}), otherwise (standard branes) we omit it.}\label{tab:branesmaximal1}
\end{table}
\begin{table}[h!]
\renewcommand{\arraystretch}{1.8}
\begin{center}
\resizebox{\textwidth}{!}{
\begin{tabular}{|c|c|c|c|c|c|c|c|c|c|}
\hline
\bf{Algebra}&$\mathbf{A_{n}}$&$\mathbf{B_{n}}$&$\mathbf{C_{n}}$&$\mathbf{D_{n}}$&$\mathbf{G_{2}}$&$\mathbf{F_{4}}$&$\mathbf{E_{6}}$&$\mathbf{E_{7}}$&$\mathbf{E_{8}}$\\ \hline \hline
\bf{dim}$\mathbf{W}$&$(n+1)!$&$2^{n}n!$&$2^{n}n!$&$2^{n-1}n!$&12&1152&72$\times$ 6!&72$\times$ 8!&192$\times$ 10!\\
\bf{dim}$\ \mathbf{\mathfrak{g}}$&$n^2+2n$&$n(2n+1)$&$n(2n+1)$&$n(2n-1)$&14&52&78&133&248\\
$|\mathbf{\Phi^{+}}|$&$\dfrac{n(n+1)}{2}$&$n^2$&$n^2$&$n(n-1)$&6&24&36&63&120\\
$|\mathbf{\Delta}|$&$n$&$n$&$n$&$n$&$2$&$4$&$6$&$7$&$8$\\
{\bf h}&$n+1$&$2n$&$2n$&$2n-2$&6&12&12&18&30\\
\hline
\end{tabular}
}
\end{center}
\caption{Some relevant features of the Lie algebras are listed: dimension, rank, number of positive roots, Coxeter number $h$ and  dimension of the corresponding Weyl group.}\label{tab:weyldim1}
\end{table}
\FloatBarrier
For convenience we report the dimension, rank, number of positive roots and  dimension of the Weyl group for the Lie algebras in \cref{tab:weyldim1}. By looking at \cref{eq:relation0,tab:weyldim1} it is immediate to recognize the following formulae
\begin{subequations}
 \begin{flalign}
 &N^{d}_{2\text{-}brane}=\frac{2^{9-d}}{11-d}\ N^d_{1\text{-}brane}\\
 &N^{d}_{3\text{-}brane}=2^{8-d}\ N^d_{1\text{-}brane}\\
 &N^{d}_{4^+\text{-}brane}=2^{9-d}\ N^d_{1\text{-}brane},
 \end{flalign}\label{eq:ranksrule}
 \end{subequations}
relating the number of different rank solutions. Moreover \cref{eq:relation0} induce also the following relations 
\begin{subequations}
 \begin{flalign}
 &N^{d+1}_{1\text{-}brane}=2(10-d)\ \frac{N^d_{1\text{-}brane}}{N^d_{0\text{-}brane}}\\
 &N^{d+1}_{2\text{-}brane}=(11-d)\ \frac{N^d_{2\text{-}brane}}{N^d_{0\text{-}brane}}\\
 &N^{d+1}_{3\text{-}brane}=(10-d)\ \frac{N^d_{3\text{-}brane}}{N^d_{0\text{-}brane}}\\
 &N^{d+1}_{4^+\text{-}brane}=(10-d)\ \frac{N^d_{4^+\text{-}brane}}{N^d_{0\text{-}brane}}
\end{flalign}\label{eq:upliftsrule}
\end{subequations}
characterizing uplift/compactification behaviors of half-supersymmetric solutions, where $4^+$ means solutions of rank four and higher, with the exception for the five-brane case understood.


\FloatBarrier
\section{Branes ad Polytopes}\label{sec:branesandpolytopes}
\FloatBarrier
In the previous section we have defined a set of algebraic relations encoding the number of half-supersymmetric solutions in maximal supergravity theories. These rules were obtained by applying some general results on the Weyl group to the irreducible representations hosting U-duality charges. In this section we look at the general setting behind the relations of \cref{eq:relation0}. A natural identification of half-supersymmetric solutions as vertices of certain classes of uniform polytopes will emerge by this way.\\  

Let's consider a Coxeter system $(W,S)$ acting on a vector space $V$. We want to take a close look to the action of $W$ on $V$ and, to this aim, we introduce the half-spaces $A_{\alpha}$ defined by the hyperplanes $H_{\alpha}$
\begin{flalign}
 &A_{\alpha}=\{\lambda\in V\ |\ \scp{\lambda}{\alpha}>0\}
\end{flalign}
and the set 
\begin{flalign}
 &C=\bigcap_{\alpha\in\Delta}A_{\alpha}.
\end{flalign}
$C$ is called \textbf{chamber} of $W$. Its closure 
\begin{flalign}
 &D=\overline{C}=\{\lambda\in V\ |\ \scp{\alpha}{\lambda}\geqslant 0\ \forall \alpha\in \Delta\}
\end{flalign}
is the \textbf{fundamental domain} of $W$ acting on $V$. Since simple roots are linearly independent and the origin belongs to each $H_{\alpha}$ with $\alpha\in \Delta$ then by definition the fundamental domain, fixing points on the intersections of the $H_{\alpha}$, is a \textbf{simplex}. Any $\mu\in V$ is Weyl conjugate to some $\lambda\in D$. The union of the images of the chambers under the action of the Coxeter group constitutes the \textbf{Tits cone}, 
\begin{flalign}
 &X=\bigcup_{w\in W}wC.
\end{flalign}
The projective space built from the Tits cone defines the \textbf{Coxeter complex}
\begin{flalign}
 &\mathcal{C}=(X/ \{0\})/\mathds{R}^{+}.
\end{flalign}
$\mathcal{C}$ is an \textbf{abstract simplicial complex}. We recall that an abstract simplicial complex $\mathcal{C}$ is a family of non-empty sets such that, for every $Y\subseteq \mathcal{C}$, any non-empty subset $X\subseteq Y$ is also in $\mathcal{C}$. The vertices of the abstract simplicial complex are in correspondence with $wW_{I}$ when $I$ is maximal in $S$, namely when $I$ contains all but one simple reflections in $S$. Subsets of the abstract simplicial complex are called \textbf{faces}.\\

We could further refine the description of the fundamental domain by taking a parabolic subgroup $W_{I}$ of $W$ and defining
\begin{flalign}
 &C_{I}=\{\lambda\in D\ |\ \scp{\lambda}{\alpha}=0\ \forall \alpha\in\Delta_{I},\ \scp{\lambda}{\alpha}>0\ \forall \alpha\in\Delta/\Delta_{I} \}.
\end{flalign}
The $C_{I}$'s partition $D$. If the Coxeter system is the Weyl group of a Lie algebra $\mathfrak{g}$ and $V$ is an irreducible representation then we identify $D$ as the set of dominant weights in the representation, while the $C_{I}$, depending on the subset $I\subseteq S$, could be different subsets of $D$. The isotropy group of $C_{I}$ is the parabolic subgroup $W_{I}$ and furthermore $wC_{I}\cap w'C_{I}=\emptyset$ if there is no $u\in W_{I}$ such that $w=w'u$, i.e. if $w$ and $w'$ do not belong to the same left coset $W/W_{I}$. $wC_{I}$ are called \textbf{facets} of type $I$. Collecting all the $wC_{I}$ for $w\in W$ and $I\subset S$ we get the \textbf{Coxeter complex}\cite{humphreys1990}
\begin{flalign}
 &\mathcal{C}=\bigcup_{\substack{w\in W\\ I\subset S}}wC_{I};
\end{flalign}
In the next when the type is not specified we use the word \textbf{facets} to denote the maximal subsets of the abstract simplicial complex, i.e. faces not contained in any other face.\\

In order to visualize the description above let's consider the Coxeter group of the Lie algebra $A_{3}$. In \cref{fig:CoxetexcomplexA3} we sketch the six walls of $A_{3}$ intersecting the unit sphere. These triangulates the sphere
delimiting twenty-four chambers, whose closures correspond to the fundamental domain and its Weyl-equivalent counterparts; these are spherical simplices. The Tits cone is built as the union of all the chambers. The intersection with the unit 2-sphere constitutes the Coxeter complex that, in this case, turns out to be a simplicial complex. The points identified by the intersection of two walls and the sphere could be seen as vertices of a p-polytope (a polyhedron in this case), a convex hull of p points, with the symmetry of the Coxeter group; it is drawn in \cref{fig:CoxetexcomplexA3b}.
\begin{figure}[h!]
\centering
\subfloat[The reflection planes of the Coxeter system of $A_{3}$ and the unit sphere. The fundamental domain is the region delimited by three walls.]{\label{fig:CoxetexcomplexA3a}
\centering
\includegraphics[height=6cm]{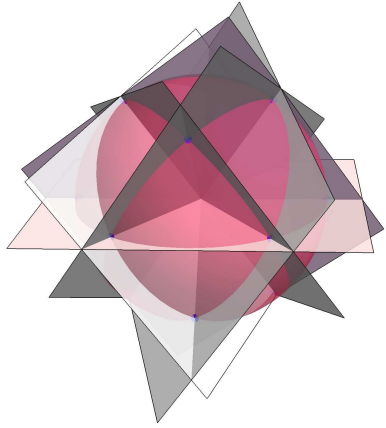}
 }\hspace{2cm}
 \subfloat[The convex hull of the points on the unit sphere identified by the Coxeter complex.]{\label{fig:CoxetexcomplexA3b}
 \centering
 \includegraphics[height=4.5cm]{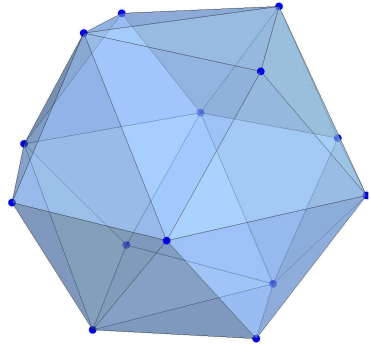}
 }
\caption{Coxeter complex of $A_{3}$ and the corresponding polytope.}\label{fig:CoxetexcomplexA3}
\end{figure} 
\FloatBarrier
We are interested in specifying the general construction presented above to half-supersymmetric branes in maximal supergravity theories and define their geometric realization within the corresponding Coxeter complex. We consider a U-duality brane representation and we take the highest weight $\Lambda$ and its isotropy group $W_{I_{0}^{\Lambda}}$ as described in \cref{sec:braneE11Weyl}. Our $C_{I_{0}^{\Lambda}}$ consists just in the highest weight itself. The intersection of the highest weight and the other longest weights of the representation, each corresponding to an half-supersymmetric solution, with the Coxeter complex identifies the vertices of a polytope, each lying on a type I facets. A vertex lying on the intersection of all but one reflection planes $H_{\alpha}$ (fixing point on the unit sphere) overlaps a point in the Coxeter complex; if we remove one hyperplane it belongs to an edge, a 1-face. Removing a further hyperplane the point will lie on a 2-face and so on. This means that if $I\subset S$ is maximal the vertices of the polytope identified by brane states coincide with the vertices of the Coxeter complex.\\
Two clarifying examples are given by the representations {\bf 4} and {\bf 20} of $A_{3}$, whose highest weights have Dynkin labels
\begin{flalign*}
 &\boxed{1\ 0\ 0}
\end{flalign*}
and
\begin{flalign*}
 &\boxed{1\ 1\ 0}
\end{flalign*}
respectively.
%

%
%
%
%
%
The polytopes corresponding to the outer Weyl orbit of these representation, i.e. the orbits of the longest weights under the action of the Weyl group, could be generated starting from the highest weight and reflecting it trough the walls $H_{\alpha}$. By this way one gets the longest weights in the representation, that are four in the {\bf 4} and twelve in the {\bf 20}. The resulting polyhedra are shown in \cref{fig:4ofA3polytope} and \cref{fig:20ofA3polytope} in purple, inside the Coxeter complex, in blue, and they correspond to a tetrahedron and a truncated tetrahedron.
\begin{figure}[h!]
\centering
\subfloat[Tetrahedron inside the Coxeter complex corresponding to the weights in the representation {\bf 4} of $A_{3}$  with highest weight $\boxed{1\ 0\ 0}$.]{\label{fig:4ofA3polytope}
\includegraphics[height=5.5cm]{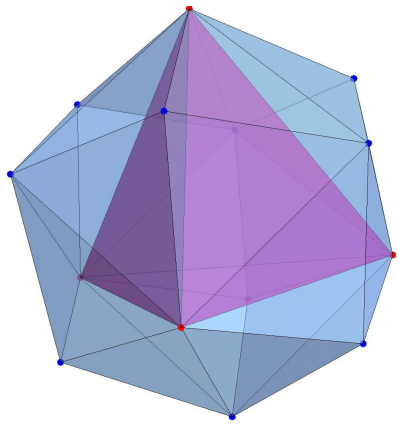}
 }\hspace{1cm}
 \subfloat[Octahedron inside the Coxeter complex corresponding to the weights in the representation {\bf 6} of $A_{3}$  with highest weight $\boxed{0\ 1\ 0}$.]{\label{fig:6ofA3polytope}
\includegraphics[height=5.5cm]{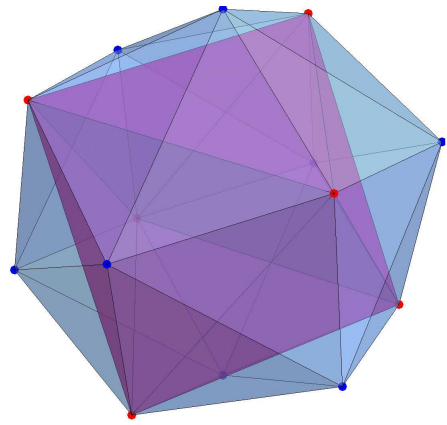}
}\\
\subfloat[Truncated tetrahedron  inside the Coxeter complex corresponding to the longest weights in the representation {\bf 20} of $A_{3}$ with highest weight $\boxed{1\ 1\ 0}$.]{\label{fig:20ofA3polytope}
 \includegraphics[height=5.5cm]{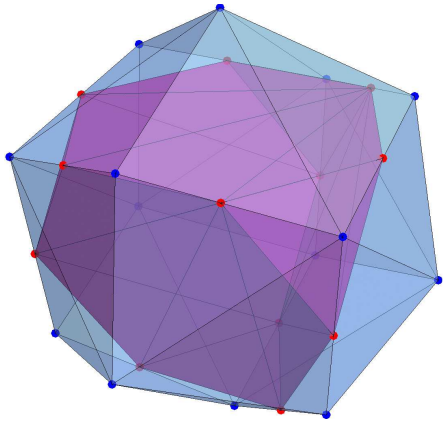}
}\hspace{1cm}
\subfloat[Cuboctahedron inside the Coxeter complex corresponding to the roots in the adjoint representation {\bf 15} of $A_{3}$  with highest root $\boxed{1\ 0\ 1}$.]{\label{fig:15ofA3polytope}
\includegraphics[height=5.5cm]{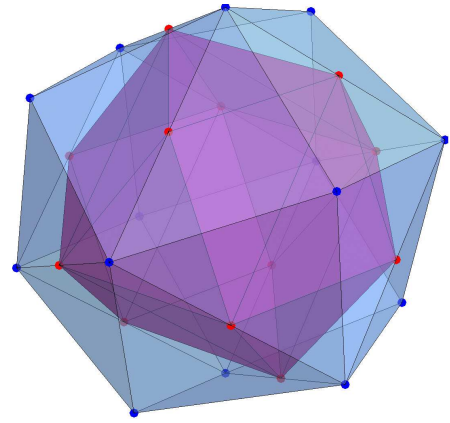}
 }\\
\subfloat[Truncated octahedron inside the Coxeter complex corresponding to the longest weights in the representation {\bf 64} of $A_{3}$  with highest weight $\boxed{1\ 1\ 1}$.]{\label{fig:64ofA3polytope}
\includegraphics[height=5.5cm]{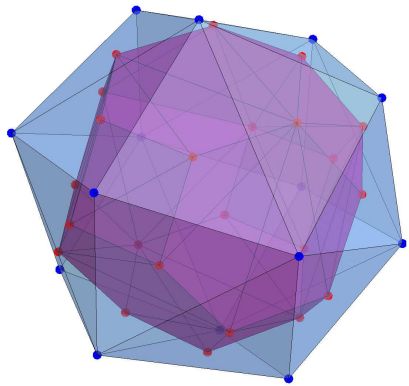}
}
\caption{Polytopes associated with the Weyl group of $A_3$ visualized inside the Coxeter complex.}
\end{figure} 
\FloatBarrier
We note also that, while the vertices of the tetrahedron, associated with the {\bf 4}, overlap four vertices of the Coxeter complex the vertices of the {\bf 20} lie on its edges. This is due to the fact that in the former case the isotropy group for the highest weight corresponds to a maximal I, it is the Coxeter group associated with the subalgebra $A_{2}$ made by the simple roots $\alpha_2$ and $\alpha_3$, while in the latter case this is not true since the isotropy group is the Weyl group of an $A_{1}$ subalgebra.\\

The link between Coxeter groups, polytopes and weights in our examples is quite general. Any polytope with pure reflectional symmetry could be represented by a Coxeter diagram with additional informations. To do this one should fix a generator point and reflect it trough the hyperplanes $H_{\alpha}$ corresponding to each node. The generator point could vary and, to identify it, the nodes in the Coxeter diagram are divided into active and inactive nodes \cite{Coxeter1985}. Active nodes are signaled by a ring in the Coxeter diagram. A node is inactive if the generator point is invariant under the reflection with respect to the corresponding hyperplane, meaning it lies on the hyperplane itself, it is active if it is not invariant, \cref{fig:activeincativenodes}.
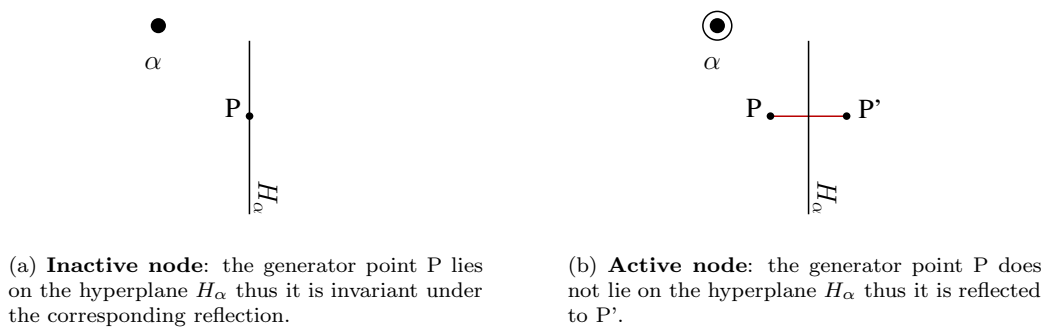
\begin{figure}[h!]
\centering
\subfloat[\textbf{Inactive node}: the generator point P lies on the hyperplane $H_{\alpha}$ thus it is invariant under the corresponding reflection.]{
\centering
\scalebox{0.5} 
{
\begin{pspicture}(-2,-1.6105468)(9.772969,3.6105468)
\psline[linewidth=0.04cm](3.981875,2.990547)(3.981875,-1.6094531)
\psdots[dotsize=0.4](1.581875,3.3905468)
\usefont{T1}{ppl}{m}{n}
\rput{-90.0}(5.6105466,3.1585155){\rput(4.3264065,-1.1994531){\huge $H_{\alpha}$}}
\usefont{T1}{ppl}{m}{n}
\rput(1.4464062,2.4005468){\huge $\alpha$}
\psdots[dotsize=0.2](3.981875,0.9905469)
\usefont{T1}{ptm}{m}{n}
\rput(3.5465624,1.2005469){\huge P}
\end{pspicture} 
}
}
\hspace{1cm}
\subfloat[\textbf{Active node}: the generator point P does not lie on the hyperplane $H_{\alpha}$ thus it is reflected to P'.]{
\centering 
\scalebox{0.5} 
{
\begin{pspicture}(-2,-1.700547)(9.7671876,3.700547)
\definecolor{color696}{rgb}{0.7215686274509804,0.00784313725490196,0.00784313725490196}
\psline[linewidth=0.04cm,linecolor=color696](2.981875,0.90054685)(4.981875,0.90054685)
\psdots[dotsize=0.2](4.981875,0.90054685)
\psdots[dotsize=0.2](2.981875,0.90054685)
\pscircle[linewidth=0.04,dimen=outer](1.581875,3.300547){0.4}
\psline[linewidth=0.04cm](3.981875,2.9005468)(3.981875,-1.6994531)
\psdots[dotsize=0.4](1.581875,3.300547)
\usefont{T1}{ppl}{m}{n}
\rput(1.4464062,2.3105469){\huge $\alpha$}
\usefont{T1}{ppl}{m}{n}
\rput{-90.0}(5.7005467,3.0685153){\rput(4.3264065,-1.2894531){\huge $H_{\alpha}$}}
\usefont{T1}{ptm}{m}{n}
\rput(2.5465624,1.1105468){\huge P}
\usefont{T1}{ptm}{m}{n}
\rput(5.6201565,1.1105468){\huge P'}
\end{pspicture} 
}

}
\caption{Active and inactive nodes in a Coxeter diagram.}\label{fig:activeincativenodes}
\end{figure}
Thus given a Coxeter diagram with active and inactive nodes, it identifies a generator point lying on the intersection of the hyperplanes associated with inactive nodes and not lying on any hyperplane corresponding to an active node. We take the generator point equidistant from the hyperplanes corresponding to active nodes. Then the polytope is built simply reflecting the generator point recursively with respect to all the hyperplanes (active and inactive). The resulting polytope has the symmetry of the Coxeter diagram. It is clear that there could be different polytopes invariant under the same Coxeter system, defined by different set of active nodes. An example of the correspondence between Coxeter diagram and polytopes just described is sketched in \cref{fig:A2polytopes}, where we consider the Coxeter system associated with $A_{2}$. In \cref{fig:A2a} both the nodes associated with the simple roots $\alpha_{1}$ and $\alpha_{2}$ are inactive thus the generator point lies on the intersection of $H_{\alpha_1}$ and $H_{\alpha_2}$; the polytope associated with the graph is trivially a point. In \cref{fig:A2b} one node is active, $\alpha_1$, and one node is inactive, $\alpha_2$, thus $P$ lies on the hyperplane $H_{\alpha_2}$. The corresponding polytope is a triangle. In \cref{fig:A2c} both nodes are active resulting in an hexagon and, as expected, this corresponds to the diagram of the root system of $A_{2}$.
\begin{figure}[h!]
\centering
\subfloat[]{\label{fig:A2a}
\centering
\scalebox{0.4} 
{
\begin{pspicture}(1,-5.26286)(12.369062,5.26286)
\definecolor{color7877b}{rgb}{0.00392156862745098,0.00392156862745098,0.00392156862745098}
\psline[linewidth=0.04cm](1.781875,3.2810493)(3.181875,3.2810493)
\pscircle[linewidth=0.02,dimen=outer,fillstyle=solid,fillcolor=color7877b](1.781875,3.2810493){0.2}
\pscircle[linewidth=0.02,dimen=outer,fillstyle=solid,fillcolor=color7877b](2.981875,3.2810493){0.2}
\usefont{T1}{ppl}{m}{n}
\rput(1.8464062,2.2910492){\huge $\alpha_1$}
\usefont{T1}{ppl}{m}{n}
\rput(3.0464063,2.2910492){\huge $\alpha_2$}
\psline[linewidth=0.04cm](0.781875,-1.3189507)(10.381875,-1.3189507)
\usefont{T1}{ppl}{m}{n}
\rput{60.0}(5.9997787,-5.062948){\rput(7.3264065,2.6910493){\huge $H_{\alpha_2}$}}
\psline[linewidth=0.04cm](3.331875,-5.24286)(7.631875,2.2049587)
\usefont{T1}{ppl}{m}{n}
\rput(10.526406,-0.9089506){\huge $H_{\alpha_1}$}
\psdots[dotsize=0.3](5.581875,-1.3189507)
\usefont{T1}{ppl}{m}{n}
\rput(5.7465625,-1.9089507){\huge P}
\end{pspicture} 
}
}
\subfloat[]{\label{fig:A2b}
\centering
\scalebox{0.4} 
{
\begin{pspicture}(1,-5.26286)(12.369062,5.26286)
\definecolor{color7877b}{rgb}{0.00392156862745098,0.00392156862745098,0.00392156862745098}
\definecolor{color47b}{rgb}{0.9647058823529412,0.00392156862745098,0.00392156862745098}
\pspolygon[linewidth=0.02,fillstyle=solid,fillcolor=red217!30](2.981875,-1.3189507)(6.981875,1.0810493)(6.981875,-3.7189507)
\pscircle[linewidth=0.04,dimen=outer](1.781875,3.2810493){0.4}
\psline[linewidth=0.04cm](1.781875,3.2810493)(3.181875,3.2810493)
\pscircle[linewidth=0.02,dimen=outer,fillstyle=solid,fillcolor=color7877b](1.781875,3.2810493){0.2}
\pscircle[linewidth=0.02,dimen=outer,fillstyle=solid,fillcolor=color7877b](2.981875,3.2810493){0.2}
\usefont{T1}{ppl}{m}{n}
\rput(1.8464062,2.2910492){\huge $\alpha_1$}
\usefont{T1}{ppl}{m}{n}
\rput(3.0464063,2.2910492){\huge $\alpha_2$}
\psline[linewidth=0.04cm](0.781875,-1.3189507)(10.381875,-1.3189507)
\usefont{T1}{ppl}{m}{n}
\rput{60.0}(5.9997787,-5.062948){\rput(7.3264065,2.6910493){\huge $H_{\alpha_2}$}}
\psline[linewidth=0.04cm](3.331875,-5.24286)(7.631875,2.2049587)
\usefont{T1}{ppl}{m}{n}
\rput(10.526406,-0.9089506){\huge $H_{\alpha_1}$}
\psdots[dotsize=0.3](6.981875,-3.7189507)
\usefont{T1}{ppl}{m}{n}
\rput(7.5465624,0.69104934){\huge P}
\psdots[dotsize=0.3](2.981875,-1.3189507)
\psdots[dotsize=0.3](6.981875,1.0810493)
\end{pspicture} 
}
}
\subfloat[]{\label{fig:A2c}
\centering
\scalebox{0.4} 
{
\begin{pspicture}(1,-5.175905)(12.369062,5.375905)
\definecolor{color7877b}{rgb}{0.00392156862745098,0.00392156862745098,0.00392156862745098}
\definecolor{color7867b}{rgb}{0.4666666666666667,0.7803921568627451,0.9098039215686274}
\pspolygon[linewidth=0.02,fillstyle=solid,fillcolor=boh!30](5.581875,2.7940948)(9.181875,0.79409474)(9.181875,-3.2059052)(5.581875,-5.2059054)(1.981875,-3.2059052)(1.981875,0.79409474)
\pscircle[linewidth=0.04,dimen=outer](1.781875,3.3940947){0.4}
\pscircle[linewidth=0.04,dimen=outer](2.981875,3.3940947){0.4}
\psdots[dotsize=0.3](1.981875,0.79409474)
\usefont{T1}{ppl}{m}{n}
\rput(9.746562,0.80409473){\huge P}
\psline[linewidth=0.04cm](1.781875,3.3940947)(3.181875,3.3940947)
\pscircle[linewidth=0.02,dimen=outer,fillstyle=solid,fillcolor=color7877b](1.781875,3.3940947){0.2}
\pscircle[linewidth=0.02,dimen=outer,fillstyle=solid,fillcolor=color7877b](2.981875,3.3940947){0.2}
\usefont{T1}{ppl}{m}{n}
\rput(1.8464062,2.4040947){\huge $\alpha_1$}
\usefont{T1}{ppl}{m}{n}
\rput(3.0464063,2.4040947){\huge $\alpha_2$}
\psline[linewidth=0.04cm](0.781875,-1.2059053)(10.381875,-1.2059053)
\usefont{T1}{ppl}{m}{n}
\rput{60.0}(6.097679,-5.0064254){\rput(7.3264065,2.8040948){\huge $H_{\alpha_2}$}}
\psline[linewidth=0.04cm](3.331875,-5.1298146)(7.631875,2.318004)
\usefont{T1}{ppl}{m}{n}
\rput(10.526406,-0.7959053){\huge $H_{\alpha_1}$}
\psdots[dotsize=0.3](9.181875,-3.2059052)
\psdots[dotsize=0.3](9.181875,0.79409474)
\psdots[dotsize=0.3](1.981875,-3.2059052)
\psdots[dotsize=0.3](5.581875,2.7940948)
\psdots[dotsize=0.3](5.581875,-5.2059054)
\end{pspicture} 
}
}\caption{Polytopes corresponding to  $A_{2}$ Coxeter system and the their Coxeter graph.}\label{fig:A2polytopes}
\end{figure}
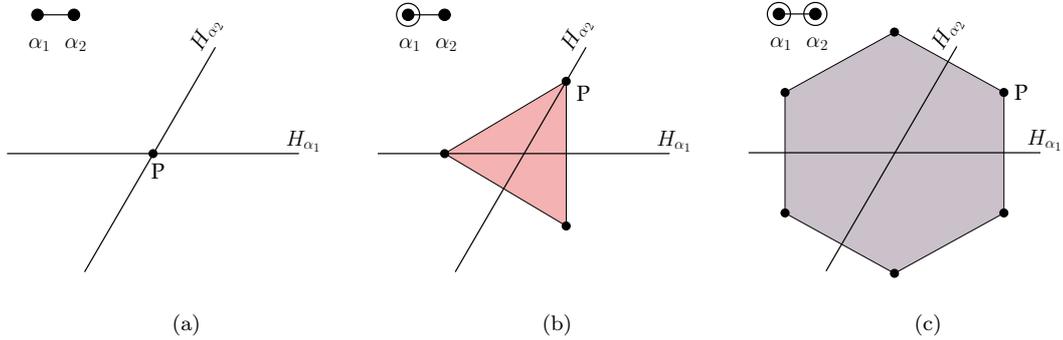
\FloatBarrier
We could associate a polytope to the longest weights of a representation by taking the highest weight as generator point and the active nodes as the nodes corresponding to its non-zero Dynkin labels. In the coset describing its orbit under the Weyl group $W/(W_{I_{0}^{\Lambda}})^N$, $W$ is the invariance group of the polytope while $W_{I_{0}^{\Lambda}}$ is the invariance group of the generator point. Thus the polyhedron associated with the longest weights in the {\bf 4} and {\bf 20} of $A_{3}$ could be conveniently represented by the following diagrams
\begin{figure}[h!]
\centering
\subfloat[Tetrahedron]{
\scalebox{0.5} 
{
\begin{pspicture}(-2,-0.4)(5.02,0.4)
\definecolor{color173b}{rgb}{0.00392156862745098,0.00392156862745098,0.00392156862745098}
\pscircle[linewidth=0.04,dimen=outer](0.4,0.0){0.4}
\psline[linewidth=0.04cm](0.4,0.0)(3.0,0.0)
\pscircle[linewidth=0.02,dimen=outer,fillstyle=solid,fillcolor=color173b](1.6,0.0){0.2}
\pscircle[linewidth=0.02,dimen=outer,fillstyle=solid,fillcolor=color173b](0.4,0.0){0.2}
\pscircle[linewidth=0.02,dimen=outer,fillstyle=solid,fillcolor=color173b](2.8,0.0){0.2}
\end{pspicture} 
}

}\hspace{1cm} 
\subfloat[Truncated tetrahedron]{
\scalebox{0.5} 
{
\begin{pspicture}(-2,-0.4)(5.02,0.4)
\definecolor{color173b}{rgb}{0.00392156862745098,0.00392156862745098,0.00392156862745098}
\pscircle[linewidth=0.04,dimen=outer](0.4,0.0){0.4}
\psline[linewidth=0.04cm](0.4,0.0)(3.0,0.0)
\pscircle[linewidth=0.02,dimen=outer,fillstyle=solid,fillcolor=color173b](1.6,0.0){0.2}
\pscircle[linewidth=0.02,dimen=outer,fillstyle=solid,fillcolor=color173b](0.4,0.0){0.2}
\pscircle[linewidth=0.02,dimen=outer,fillstyle=solid,fillcolor=color173b](2.8,0.0){0.2}
\pscircle[linewidth=0.04,dimen=outer](1.6,0.0){0.4}
\end{pspicture} 
}

} 
\end{figure}\\ \FloatBarrier
We complete the analysis of the polyhedra with symmetry of the Coxeter group of $A_{3}$ taking also those corresponding to the outer Weyl orbits of the representations {\bf 6}, {\bf 15} (the adjoint) and {\bf 64}. They are drawn in \cref{fig:6ofA3polytope,fig:15ofA3polytope,fig:64ofA3polytope} and the corresponding Coxeter diagrams are listed in \cref{tab:A3polytope}. It is interesting to note that the vertices of the {\bf 64}, having minimal isotropy group, lie on the face of the Coxeter complex. The polytopes listed in \cref{tab:A3polytope} exhaust all the possibilities for $A_{3}$.
\begin{table}[h!]
\renewcommand{\arraystretch}{2}
\begin{center}
\begin{tabular}{|c|c|c|c|c|c|c|}
\hline
{\bf rep}&\textbf{Dynkin labels}&\bf{Coxeter diagram}&\textbf{Polytope}& \textbf{V}&\bf{E}&\bf{F} \\
\hline\hline
{\bf 4}&$\boxed{1\ 0\ 0}$&
\scalebox{0.5} 
{
\begin{pspicture}(-0,-0.4)(3.02,0.4)
\definecolor{color173b}{rgb}{0.00392156862745098,0.00392156862745098,0.00392156862745098}
\pscircle[linewidth=0.04,dimen=outer](0.4,0.0){0.4}
\psline[linewidth=0.04cm](0.4,0.0)(3.0,0.0)
\pscircle[linewidth=0.02,dimen=outer,fillstyle=solid,fillcolor=color173b](1.6,0.0){0.2}
\pscircle[linewidth=0.02,dimen=outer,fillstyle=solid,fillcolor=color173b](0.4,0.0){0.2}
\pscircle[linewidth=0.02,dimen=outer,fillstyle=solid,fillcolor=color173b](2.8,0.0){0.2}
\end{pspicture} 
}
& tetrahedron&4&6&4\\
{\bf 6}&$\boxed{0\ 1\ 0}$&
\scalebox{0.5} 
{
\begin{pspicture}(-0,-0.4)(3.02,0.4)
\definecolor{color173b}{rgb}{0.00392156862745098,0.00392156862745098,0.00392156862745098}
\pscircle[linewidth=0.04,dimen=outer](1.6,0.0){0.4}
\psline[linewidth=0.04cm](0.4,0.0)(3.0,0.0)
\pscircle[linewidth=0.02,dimen=outer,fillstyle=solid,fillcolor=color173b](1.6,0.0){0.2}
\pscircle[linewidth=0.02,dimen=outer,fillstyle=solid,fillcolor=color173b](0.4,0.0){0.2}
\pscircle[linewidth=0.02,dimen=outer,fillstyle=solid,fillcolor=color173b](2.8,0.0){0.2}
\end{pspicture} 
}
&octahedron &6&12&8\\
{\bf 20}&$\boxed{1\ 1\ 0}$&
\scalebox{0.5} 
{
\begin{pspicture}(-0,-0.4)(3.02,0.4)
\definecolor{color173b}{rgb}{0.00392156862745098,0.00392156862745098,0.00392156862745098}
\pscircle[linewidth=0.04,dimen=outer](0.4,0.0){0.4}
\psline[linewidth=0.04cm](0.4,0.0)(3.0,0.0)
\pscircle[linewidth=0.02,dimen=outer,fillstyle=solid,fillcolor=color173b](1.6,0.0){0.2}
\pscircle[linewidth=0.02,dimen=outer,fillstyle=solid,fillcolor=color173b](0.4,0.0){0.2}
\pscircle[linewidth=0.02,dimen=outer,fillstyle=solid,fillcolor=color173b](2.8,0.0){0.2}
\pscircle[linewidth=0.04,dimen=outer](1.6,0.0){0.4}
\end{pspicture} 
}
&truncated tetrahedron &12&18&8\\
{\bf 15}&$\boxed{1\ 0\ 1}$&
\scalebox{0.5} 
{
\begin{pspicture}(-0,-0.4)(3.02,0.4)
\definecolor{color173b}{rgb}{0.00392156862745098,0.00392156862745098,0.00392156862745098}
\pscircle[linewidth=0.04,dimen=outer](0.4,0.0){0.4}
\psline[linewidth=0.04cm](0.4,0.0)(3.0,0.0)
\pscircle[linewidth=0.02,dimen=outer,fillstyle=solid,fillcolor=color173b](1.6,0.0){0.2}
\pscircle[linewidth=0.02,dimen=outer,fillstyle=solid,fillcolor=color173b](0.4,0.0){0.2}
\pscircle[linewidth=0.02,dimen=outer,fillstyle=solid,fillcolor=color173b](2.8,0.0){0.2}
\pscircle[linewidth=0.04,dimen=outer](2.8,0.0){0.4}
\end{pspicture} 
}
&cuboctahedron &12&24&14\\
{\bf 64}&$\boxed{1\ 1\ 1}$&
\scalebox{0.5} 
{
\begin{pspicture}(-0,-0.4)(3.02,0.4)
\definecolor{color173b}{rgb}{0.00392156862745098,0.00392156862745098,0.00392156862745098}
\pscircle[linewidth=0.04,dimen=outer](0.4,0.0){0.4}
\psline[linewidth=0.04cm](0.4,0.0)(3.0,0.0)
\pscircle[linewidth=0.02,dimen=outer,fillstyle=solid,fillcolor=color173b](1.6,0.0){0.2}
\pscircle[linewidth=0.02,dimen=outer,fillstyle=solid,fillcolor=color173b](0.4,0.0){0.2}
\pscircle[linewidth=0.02,dimen=outer,fillstyle=solid,fillcolor=color173b](2.8,0.0){0.2}
\pscircle[linewidth=0.04,dimen=outer](1.6,0.0){0.4}
\pscircle[linewidth=0.04,dimen=outer](2.8,0.0){0.4}
\end{pspicture} 
}
&truncated octahedron &24&36&14\\ \hline
\end{tabular}
\end{center}
\caption{Polyhedra with symmetry of the Coxeter group of $A_{3}$. In the last three columns we list the number of vertices, edges and faces. In the first two columns we report the representations associated with the polytope. It is clear that there could be more representations whose outer Weyl orbits correspond to the same polytope; for example the outer Weyl orbit of the representation {\bf 10} with highest weight $\boxed{2\ 0\ 0}$ corresponds to the same polytope of the {\bf 4}. It is the isotropy group that matters, i.e. the number and position of zeros in the Dynkin labels of the highest weight.}\label{tab:A3polytope}
\end{table}\\

At this point we have all the ingredients to start a systematic analysis of the polytopes associated with half-supersymmetric branes in maximal supergravity theories, guided by  \cref{weightsd3,weightsd4,weightsd5,weightsd6,weightsd7,weightsd8} and the relations of \cref{eq:relation0}. We limit our attention to maximal theories from three to eight dimensions, the nine dimensional case being trivial.
\paragraph{0-branes}
The first case we analyze is the case of 0-branes. For 0-branes in $d$ dimensions the isotropy group of the highest weight is the Weyl group of $E_{10-d}$, \cref{eq:0brane}. Since this corresponds to a maximal parabolic subgroup of $W_{E_{11-d}}$ we immediately recognize that the vertices of the corresponding polytope overlap some vertices of the Coxeter complex. We note also that, apart from the eight dimensional case, that we discuss separately, each highest weight has Dynkin labels of the same form; only the first label, up to symmetries of the Dynkin diagram, is different from zero. In the eight dimensional case we have two Dynkin labels different from zero, but the U-duality algebra is not simple thus we have a non zero Dynkin label for each simple factor and, as it will be clear in a few, it still shares the general features of the other 0-brane highest weights. We list the polytopes identified by this way in \cref{tab:visual0}, where we show the dimension of the maximal supergravity theory, the U-duality group and the brane representation, the name of the polytope, the corresponding Coxeter  diagram, the number of vertices, the number and type of facets and the Petrie polygon. A \textbf{Petrie polygon} of an n-dimensional polytope is a skew polygon such that any n-1 consecutive sides, but not n, belong to a Petrie polygon of a facet \cite{coxeter1973regular}. These polygons are useful to understand the properties of higher dimensional polytopes \cite{DuVal201310}. In \cref{tab:visual0} and tables next to come the Petrie polygons are obtained as projection on the Coxeter plane associated with the Coxeter group of the U-duality group; taking a Coxeter element $w$ the Coxeter plane is the plane uniquely defined as the plane on which $w$ acts as a rotation of  $2\pi/h$, where $h$ is the Coxeter number. In the Petrie polygons yellow points have degeneracy three, orange points two and red points no degeneracy; we refer to \cref{sec:appB} for further details on Petrie polygons. The polytopes corresponding to 0-brane weights belong to the family $k_{21}$ of uniform polytopes \cite{Coxeter1988, Gosset1900}, where $k$ is related to the dimension by $k=7-d$. The name of the family is part of a general notation for $E_{n}$ group as
\begin{flalign}
 &E_{k+4}=[3^{k,1,2}].
\end{flalign}
The notation describes the Coxeter diagram, with 3 legs around a node built of k, 1 and 2 nodes and could be easily generalized to other cases. Taking [$3^{p,q,r}$] it is natural to associate to the polytopes defined by a single ring on the first node of the p, q and r legs the symbols
\begin{flalign}
 &p_{qr}\qquad q_{pr}\qquad r_{pq}
\end{flalign}
respectively. Specializing to $E_{k+4}$ this explains the name of the polytopes describing 0-branes and, as we will see, also 1- and 2-branes. The polytopes in the $k_{21}$ family just discussed and the ones we will deal with are all uniform polytopes. A \textbf{uniform polytope} is an isogonal polytope with uniform facets \cite{Coxeter1940}. A polytope is said to be \textbf{isogonal} or \textbf{vertex-transitive} if for any two vertices there is a transformation mapping the first isometrically onto the second.\\

Any $k_{21}$ polytope has vertex figure a $(k-1)_{21}$ polytope. The vertex figure of a polyhedron at vertex $v$ is the polygon with vertices the middle points along each edge ending on $v$ \cite{coxeter1973regular}. This immediately generalizes to higher dimensional polytopes. In the case of uniform polytopes it is clear that any vertex has the same vertex figure.
\begin{table}[h!]
\renewcommand{\arraystretch}{1.5}
\begin{center}
\resizebox{\textwidth}{!}{

\\
\hline

\end{tabular}
}
\end{center}
\caption{0-branes in maximal supergravity theories and corresponding polytopes. In the first two columns we list the dimension of the maximal supergravity, its U-duality group and the representation hosting the 1-forms. In the third column we show the Coxeter Dynkin diagram while $k_{21}$ and V are the name identifying the polytope and the number of its vertices respectively. In the last three columns we show the associated Petrie polygon and the number and types of facets. }\label{tab:visual0}
\end{table}

\paragraph{1-branes}
For the 1-branes the isotropy group is the Weyl group of $D_{10-d}$ \cref{eq:1brane} and the half-supersymmetric solutions could be seen as vertices of the family of uniform polytope $2_{k1}$ where again $k=7-d$. We list all of them in \cref{tab:visual1}. For the moment let's note that $2_{k1}$ polytopes have two types of facets: $2_{(k-1) 1}$ polytopes and $(k+3)$-simplexes; in order to have a comprehensive view, we will analyze this feature after we have discussed the 2-brane case also.
\begin{table}[h!]
\renewcommand{\arraystretch}{1.5}
\begin{center}
\resizebox{\textwidth}{!}{


\\
\hline

\end{tabular}
}
\end{center}
\caption{The uniform $2_{k1}$ polytopes correspond to 1-branes in maximal supergravities. In the first two columns we list the 
dimension of the maximal supergravity, its U-duality group and the representation hosting 1-branes. In the third column we show the Coxeter Dynkin diagram while $2_{k1}$ and V are the name identifying the polytope and the number of its vertices respectively. In the last three columns we show the associated Petrie polygon and the number and types of facets.
}\label{tab:visual1}
\end{table}

\paragraph{2-branes}
In the case of 2-branes the isotropy group is the Weyl group of $A_{10-d}$. The polytopes corresponding to 2-branes are listed \cref{tab:visual2}; they belong to the family of uniform polytopes $1_{k2}$ with $k$ related to the dimension by $k=7-d$. The facets of a $1_{k2}$ polytopes are $1_{(k-1) 2}$ polytopes.

\begin{table}[h!]
\renewcommand{\arraystretch}{1.5}
\begin{center}
\resizebox{\textwidth}{!}{

\\ \hline

\end{tabular}
}
\end{center}
\caption{The uniform $1_{k2}$ polytopes correspond to 2-branes in maximal supergravities. In the first two columns we list the 
dimension of the maximal supergravity, its U-duality group and the representation hosting 3-forms. In the third column we show the Coxeter Dynkin diagram while $1_{k2}$ and V are the name identifying the polytope and the number of its vertices respectively. In the last three columns we show the associated Petrie polygon and the number and types of facets.}\label{tab:visual2}

\end{table}
\FloatBarrier
\paragraph{Triality}

At this point we could make a step back to take a general picture of what we have found for 0-, 1- and 2-branes. We have discovered that these are described by the families of polytopes $k_{21}$, $2_{k1}$ and $1_{k2}$. With the notation introduced in the previous paragraph for a Coxeter system [$3^{p,q,r}$] a polytope $p_{qr}$ has facets of type $p_{q-1r}$ and $p_{qr-1}$ and their centers are the vertices of $q_{pr}$ and $r_{pq}$ polytopes respectively \cite{Coxeter1988}. This defines a triality relation between 0-, 1- and 2-branes. In particular the $k_{21}$ polytopes describing 0-branes have two types of facets, (n-1)-simplexes and (n-1)-orthoplexes. We report the definition of orthoplex and, for completeness, we recall also the definitions of simplex \cite{coxeter1973regular}. An \textbf{n-simplex} is the convex hull of n+1 points $\{v_{0},...,v_{n}\}$ such that $v_1-v_0,...,v_n-v_0$ are linearly independent. An \textbf{n-orthoplex} or \textbf{cross n-polytope}  is the n-dimensional polytope with 2n vertices with coordinate $(\pm 1,0,...,0)$ and its permutations. An orthoplex could be also defined as the closed unit ball in $\mathds{R}^{n}$ in taxicab geometry, i.e. as 
\begin{flalign}
 &B=\{x\in\mathds{R}^n\ |\ ||x||_{l_1}\leqslant 1\},
\end{flalign}
where the $l_{1}$-norm is defined by
\begin{flalign}
 &||\mathbf{v}-\mathbf{u}||_{l_1}=\sum_{i}|v_{i}-u_{i}|
\end{flalign}
for two vectors $\mathbf{u}, \mathbf{v}$ with coordinates $\mathbf{u}=(u_1,...,u_n)$ and $\mathbf{v}=(v_1,...,v_n)$. While a simplex could be seen as the higher dimensional generalization of a triangle in a two dimensional space, an orthoplex is the higher dimensional generalization of a square in two dimensions and an octahedron in three dimensions. A simplex has the symmetry of the $A_{n}$ Coxeter group while an orthoplex is invariant under the $B_{n}$ or $D_{n}$ Coxeter group.  
In the 0-brane polytopes $k_{21}$ the number of 1-branes corresponds exactly to the number of orthoplex facets while the number of 2-branes corresponds to the number of simplicial facets. The corresponding polytopes could be built as convex hull of the central point of these two types of facets.\\
Analogous considerations apply to 1-branes and 2-branes as can be seen in \cref{tab:visual1,tab:visual2}.

\paragraph{3-branes}
For the 3-branes we encounter again the family of $2_{k1}$ polytopes but in their rectified form. \textbf{Rectification} is an operation on polytopes consisting in cutting the polytope at each vertex with a plane passing trough the midpoints of edges ending on it. This exposes the vertex figure of the initial polytope and produces a polytope with a number of vertices equal to the number of edges of the starting figure; it is denoted with a prefix $r$ before the polytope name. An example of rectification applied to a cube can be seen in \cref{fig:rectifiecube}. The polytopes corresponding to 3-branes in $d$ dimensions are the rectified $2_{(7-d)1}$ polytopes appearing in the 1-brane cases and thus they could be seen also as edges of the  $2_{(7-d)1}$ polytopes with vertices corresponding to 1-branes. In \cref{tab:visual3} we show all the $r 2_{k1}$ polytopes with their main features.
\begin{figure}[h!]
\centering
\subfloat[Cube.]{\includegraphics[height=3cm]{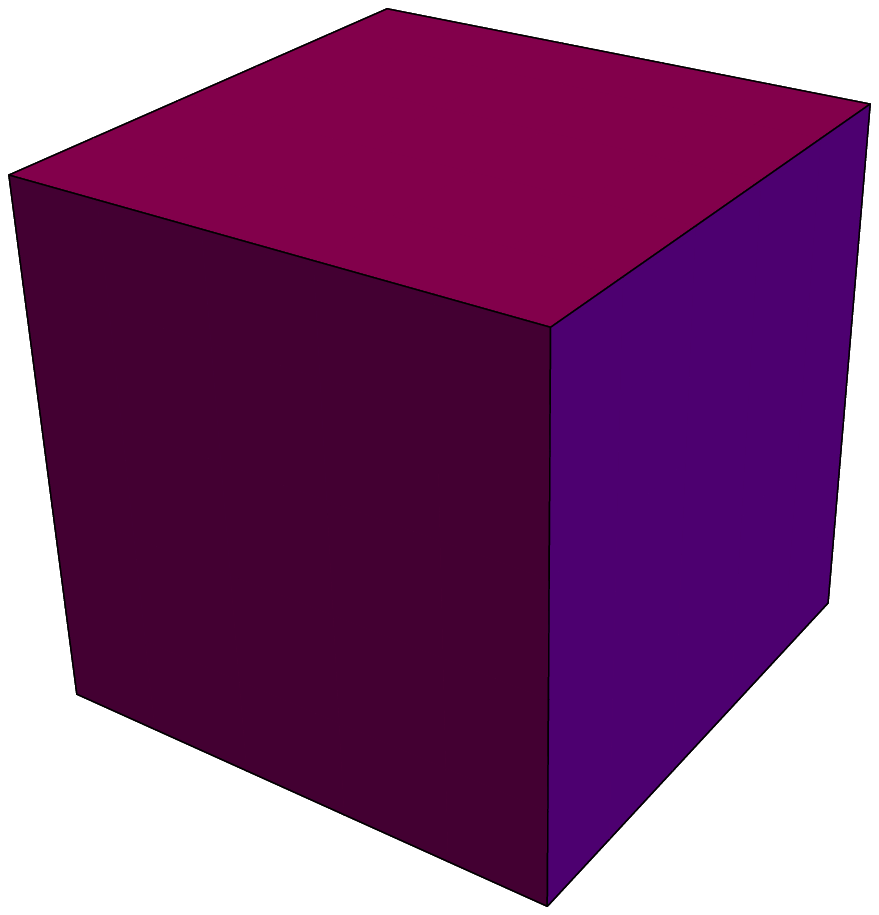}
}\hspace{1cm}
\subfloat[Rectified cube.]{\includegraphics[height=3cm]{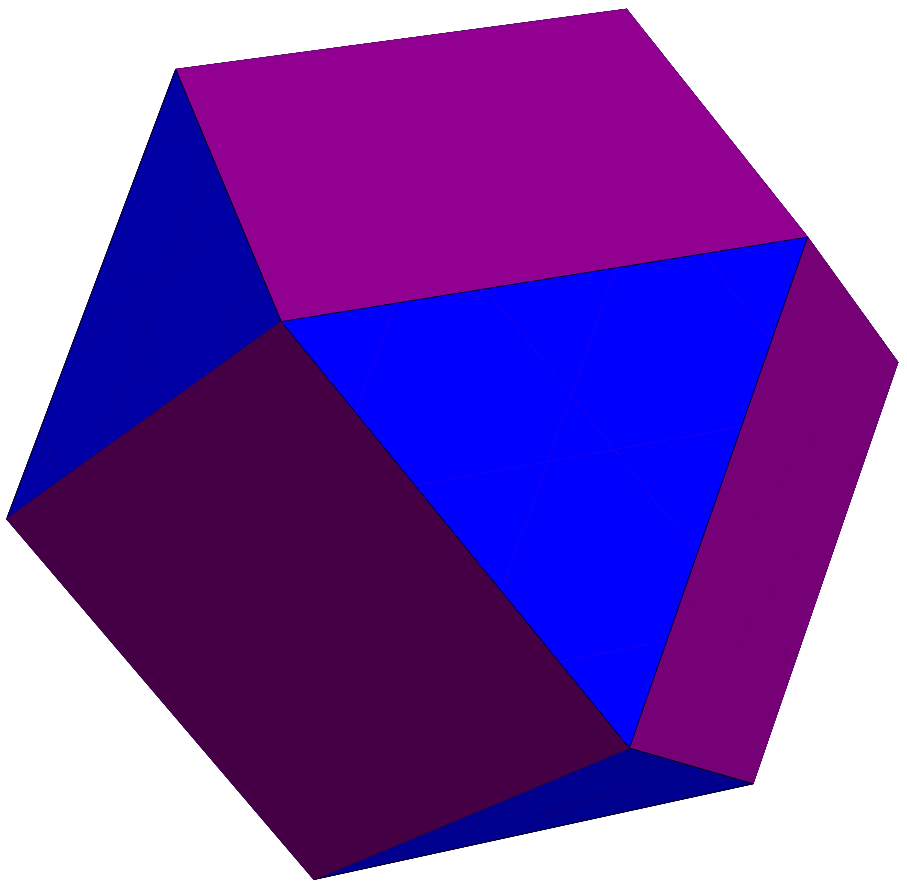}
}\caption{Cube and rectified cube.}\label{fig:rectifiecube}
\end{figure}

\begin{table}[h!]
\renewcommand{\arraystretch}{1.5}
\begin{center}
\resizebox{\textwidth}{!}{
\begin{tabular}{|c|c|c|c|c|c|}
\hline 
\multicolumn{6}{|c|}{\bf{3-Branes: rectified $\mathbf{2_{k1}}$ Polytopes}}\\
\hline
\bf{d}&\bf{G/rep}&\bf{Coxeter-Dynkin diagram}&$\mathbf{r2_{k1}}$&\bf{V}&\bf{Petrie Polygon}\\ 
\hline 
4&\begin{tabular}{c}
$E_{7}$\\
{\bf 8645}
\end{tabular}&
\scalebox{0.5} 
{
\begin{pspicture}(0,-0.9)(6.42,0.9)
\definecolor{color69b}{rgb}{0.00392156862745098,0.00392156862745098,0.00392156862745098}
\psline[linewidth=0.04cm](3.8,0.7)(3.8,-0.5)
\psline[linewidth=0.04cm](0.4,-0.5)(6.4,-0.5)
\pscircle[linewidth=0.02,dimen=outer,fillstyle=solid,fillcolor=color69b](0.2,-0.5){0.2}
\pscircle[linewidth=0.02,dimen=outer,fillstyle=solid,fillcolor=color69b](5.0,-0.5){0.2}
\pscircle[linewidth=0.02,dimen=outer,fillstyle=solid,fillcolor=color69b](3.8,-0.5){0.2}
\pscircle[linewidth=0.02,dimen=outer,fillstyle=solid,fillcolor=color69b](2.6,-0.5){0.2}
\pscircle[linewidth=0.02,dimen=outer,fillstyle=solid,fillcolor=color69b](1.4,-0.5){0.2}
\pscircle[linewidth=0.02,dimen=outer,fillstyle=solid,fillcolor=color69b](6.2,-0.5){0.2}
\pscircle[linewidth=0.02,dimen=outer,fillstyle=solid,fillcolor=color69b](3.8,0.7){0.2}
\pscircle[linewidth=0.04,dimen=outer](5.0,-0.5){0.4}
\end{pspicture} 
}

&$r2_{31}$&2016&
\begin{tabular}{c}
\\ [-0.4cm]

\includegraphics[scale=0.7]{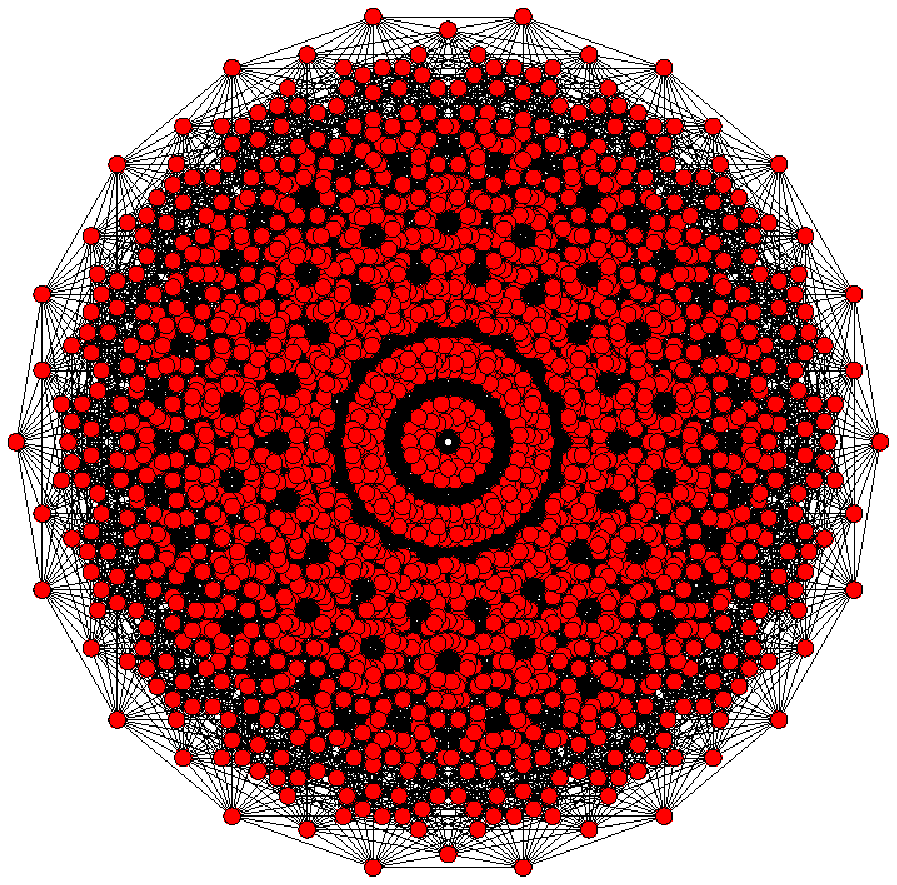}

\end{tabular}\\ \hline
5&\begin{tabular}{c}
$E_{6}$\\
{\bf 351}
\end{tabular}&
\scalebox{0.5} 
{
\begin{pspicture}(0,-0.9)(5.22,0.9)
\definecolor{color97b}{rgb}{0.00392156862745098,0.00392156862745098,0.00392156862745098}
\psline[linewidth=0.04cm](2.6,0.7)(2.6,-0.5)
\psline[linewidth=0.04cm](0.2,-0.5)(5.2,-0.5)
\pscircle[linewidth=0.02,dimen=outer,fillstyle=solid,fillcolor=color97b](3.8,-0.5){0.2}
\pscircle[linewidth=0.02,dimen=outer,fillstyle=solid,fillcolor=color97b](2.6,-0.5){0.2}
\pscircle[linewidth=0.02,dimen=outer,fillstyle=solid,fillcolor=color97b](1.4,-0.5){0.2}
\pscircle[linewidth=0.02,dimen=outer,fillstyle=solid,fillcolor=color97b](0.2,-0.5){0.2}
\pscircle[linewidth=0.02,dimen=outer,fillstyle=solid,fillcolor=color97b](5.0,-0.5){0.2}
\pscircle[linewidth=0.02,dimen=outer,fillstyle=solid,fillcolor=color97b](2.6,0.7){0.2}
\pscircle[linewidth=0.04,dimen=outer](3.8,-0.5){0.4}
\end{pspicture} 
}
&$r2_{21}$&216&
\begin{tabular}{c}
\\ [-0.4cm]
\includegraphics[scale=0.55]{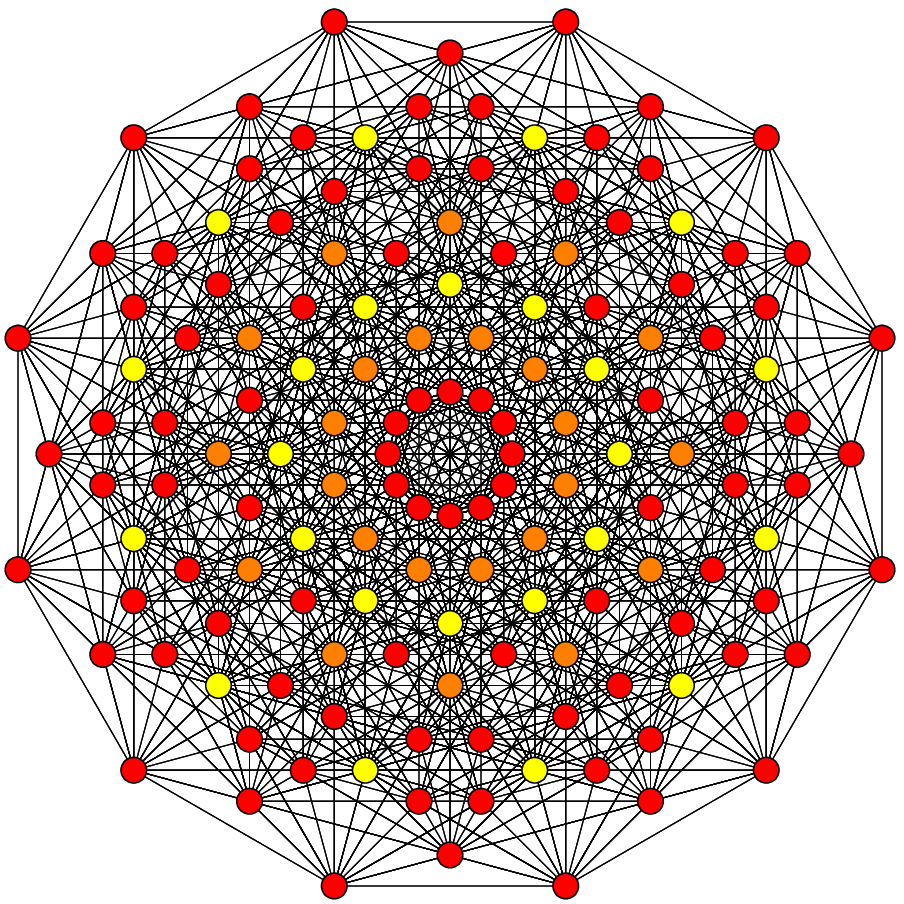}

\end{tabular}\\ \hline
6&\begin{tabular}{c}
$D_{5}$\\
{\bf 45}
\end{tabular}&
\scalebox{0.5} 
{
\begin{pspicture}(0,-0.9)(4.02,0.9)
\definecolor{color97b}{rgb}{0.00392156862745098,0.00392156862745098,0.00392156862745098}
\psline[linewidth=0.04cm](1.4,0.7)(1.4,-0.5)
\psline[linewidth=0.04cm](0.2,-0.5)(4.0,-0.5)
\pscircle[linewidth=0.02,dimen=outer,fillstyle=solid,fillcolor=color97b](2.6,-0.5){0.2}
\pscircle[linewidth=0.02,dimen=outer,fillstyle=solid,fillcolor=color97b](1.4,-0.5){0.2}
\pscircle[linewidth=0.02,dimen=outer,fillstyle=solid,fillcolor=color97b](0.2,-0.5){0.2}
\pscircle[linewidth=0.02,dimen=outer,fillstyle=solid,fillcolor=color97b](3.8,-0.5){0.2}
\pscircle[linewidth=0.02,dimen=outer,fillstyle=solid,fillcolor=color97b](1.4,0.7){0.2}
\pscircle[linewidth=0.04,dimen=outer](2.6,-0.5){0.4}
\end{pspicture} 
}

&$r2_{11}$&40&
\begin{tabular}{c}
\\ [-0.4cm]
\includegraphics[scale=0.4]{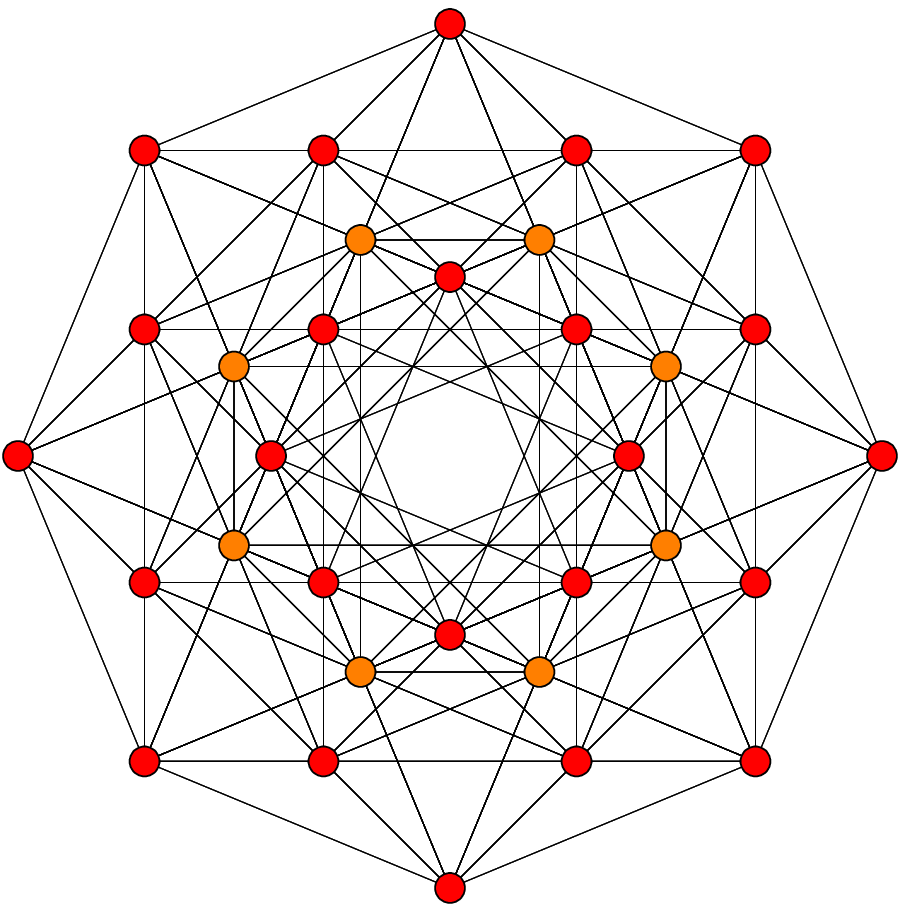}

\end{tabular}\\ \hline
7&\begin{tabular}{c}
$A_{4}$\\
{\bf 10}
\end{tabular}&
\scalebox{0.5} 
{
\begin{pspicture}(0,-0.9)(2.82,0.9)
\definecolor{color97b}{rgb}{0.00392156862745098,0.00392156862745098,0.00392156862745098}
\psline[linewidth=0.04cm](0.2,0.7)(0.2,-0.5)
\psline[linewidth=0.04cm](0.2,-0.5)(2.8,-0.5)
\pscircle[linewidth=0.02,dimen=outer,fillstyle=solid,fillcolor=color97b](1.4,-0.5){0.2}
\pscircle[linewidth=0.02,dimen=outer,fillstyle=solid,fillcolor=color97b](0.2,-0.5){0.2}
\pscircle[linewidth=0.02,dimen=outer,fillstyle=solid,fillcolor=color97b](2.6,-0.5){0.2}
\pscircle[linewidth=0.02,dimen=outer,fillstyle=solid,fillcolor=color97b](0.2,0.7){0.2}
\pscircle[linewidth=0.04,dimen=outer](1.4,-0.5){0.4}
\end{pspicture} 
}
&$r2_{01}$&10&
\begin{tabular}{c}
\\ [-0.4cm]
\includegraphics[scale=0.25]{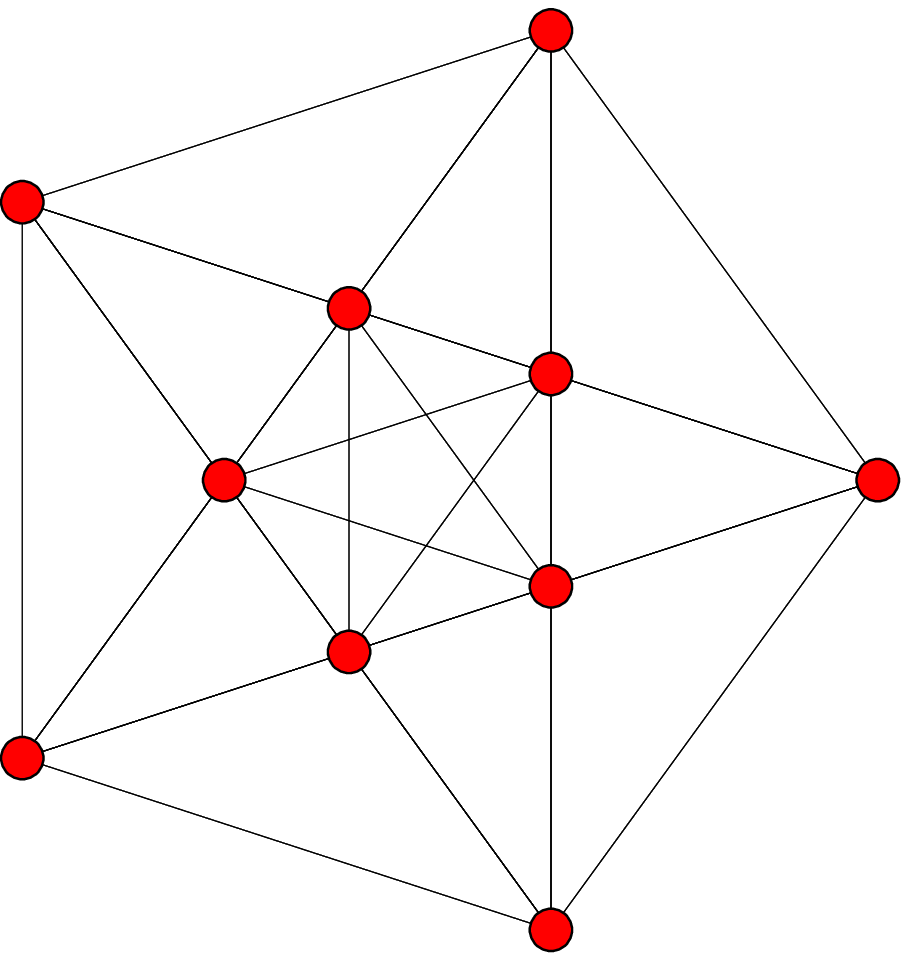}
\end{tabular}
\\ \hline
8&\begin{tabular}{c}
$A_{2}\times A_{1}$\\
($\mathbf{\overline{3}}$,{\bf 1})
\end{tabular}&
\begin{tabular}{c}
\scalebox{0.4} 
{
\begin{pspicture}(0,-0.4)(3.0,0.4)
\definecolor{color458b}{rgb}{0.00392156862745098,0.00392156862745098,0.00392156862745098}
\pscircle[linewidth=0.04,dimen=outer](1.4,0.0){0.4}
\psline[linewidth=0.04cm](0.2,0.0)(1.4,0.0)
\pscircle[linewidth=0.02,dimen=outer,fillstyle=solid,fillcolor=color458b](1.4,0.0){0.2}
\pscircle[linewidth=0.02,dimen=outer,fillstyle=solid,fillcolor=color458b](0.2,0.0){0.2}
\pscircle[linewidth=0.02,dimen=outer,fillstyle=solid,fillcolor=color458b](2.6,0.0){0.2}
\end{pspicture} 
}
\end{tabular}
&$r2_{-11}$&3&
\begin{tabular}{c}
\\ [-0.4cm]
\scalebox{0.25} 
{
\begin{pspicture}(0,-3.15)(6.8,3.5)
\pstriangle[linewidth=0.04,dimen=outer](3.4,-2.75)(6.4,5.6)
\pscircle[linewidth=0.04,dimen=outer,fillstyle=solid,fillcolor=red](0.3,-2.65){0.4}
\pscircle[linewidth=0.04,dimen=outer,fillstyle=solid,fillcolor=red](6.5,-2.65){0.4}
\pscircle[linewidth=0.04,dimen=outer,fillstyle=solid,fillcolor=red](3.4,2.65){0.4}
\end{pspicture} 
}

\end{tabular}

\\
\hline

\end{tabular}
}
\end{center}
\caption{The uniform $r2_{k1}$ polytopes correspond to 3-branes in maximal supergravities. In the first two columns we list the dimension of the maximal supergravity, its U-duality group and the representation hosting 3-forms. In the third column we show the Coxeter Dynkin diagram while $r2_{k2}$ and V are the name identifying the polytope and the number of its vertices respectively. In the last column we show the associated Petrie polygon.}\label{tab:visual3}
\end{table}

\paragraph{4-branes and beyond}
Looking at \cref{eq:relation0} we realize that for 4-brane and higher rank branes the isotropy group of the highest weight is always $W_{A_{9-d}}$, with a further class of 5-branes mimicking the situation already described for the 2-branes. The fact that 5-branes live in two representations depends on the fact that these couple both with vector and tensor multiplets. In particular 5-branes coupled to tensor multiplets obey the relations holding for the 2-branes, i.e. the one appearing in the second term of \cref{eq:5brane}, while 5-branes coupled to vector multiplets are described by the first term of the same equation. It is interesting to note that there is again a fixed scheme for non zero Dynkin labels, as can be seen in \cref{weightsd5,weightsd6,weightsd7,weightsd8,weightsd9},  but now this finds realization in a set of uniform polytopes that could not be traced back to a single family, \cref{tab:visual4}. The symbol $t_{0,3}$ appearing in \cref{tab:visual4} means that the corresponding polytope is {\bf runcinated}. Runcination is a transformation similar to rectification, where the original polytope is sliced simultaneously along faces, edges and vertices. In \cref{tab:visual4} for the 6-polytope \textit{hejack} is the Bowers acronym.\\

\begin{table}[h!]
\renewcommand{\arraystretch}{1.5}
\begin{center}
\resizebox{\textwidth}{!}{
\begin{tabular}{|c|c|c|c|c|c|}
\hline 
\multicolumn{6}{|c|}{\bf{4-Brane Polytopes}}\\
\hline
\bf{d}&\bf{G/rep}&\bf{Coxeter-Dynkin diagram}&\bf{P}&\bf{V}&\bf{Petrie Polygon}\\ 
\hline 

5&\begin{tabular}{c}
$E_{6}$\\
{\bf 1728}
\end{tabular}&
\begin{tabular}{c}
\textcolor{black!80}{demified icosiheptaheptacontidipeton (hejack)}\\
\\
\scalebox{0.5} 
{
\begin{pspicture}(0,-0.9)(5.22,0.9)
\definecolor{color97b}{rgb}{0.00392156862745098,0.00392156862745098,0.00392156862745098}
\psline[linewidth=0.04cm](2.6,0.7)(2.6,-0.5)
\psline[linewidth=0.04cm](0.2,-0.5)(5.2,-0.5)
\pscircle[linewidth=0.02,dimen=outer,fillstyle=solid,fillcolor=color97b](3.8,-0.5){0.2}
\pscircle[linewidth=0.02,dimen=outer,fillstyle=solid,fillcolor=color97b](2.6,-0.5){0.2}
\pscircle[linewidth=0.02,dimen=outer,fillstyle=solid,fillcolor=color97b](1.4,-0.5){0.2}
\pscircle[linewidth=0.02,dimen=outer,fillstyle=solid,fillcolor=color97b](0.2,-0.5){0.2}
\pscircle[linewidth=0.02,dimen=outer,fillstyle=solid,fillcolor=color97b](5.0,-0.5){0.2}
\pscircle[linewidth=0.02,dimen=outer,fillstyle=solid,fillcolor=color97b](2.6,0.7){0.2}
\pscircle[linewidth=0.04,dimen=outer](5.0,-0.5){0.4}
\pscircle[linewidth=0.04,dimen=outer](2.6,0.7){0.4}
\end{pspicture} 
}
\end{tabular}
&&432&
\begin{tabular}{c}
\\ [-0.4cm]
\includegraphics[scale=0.6]{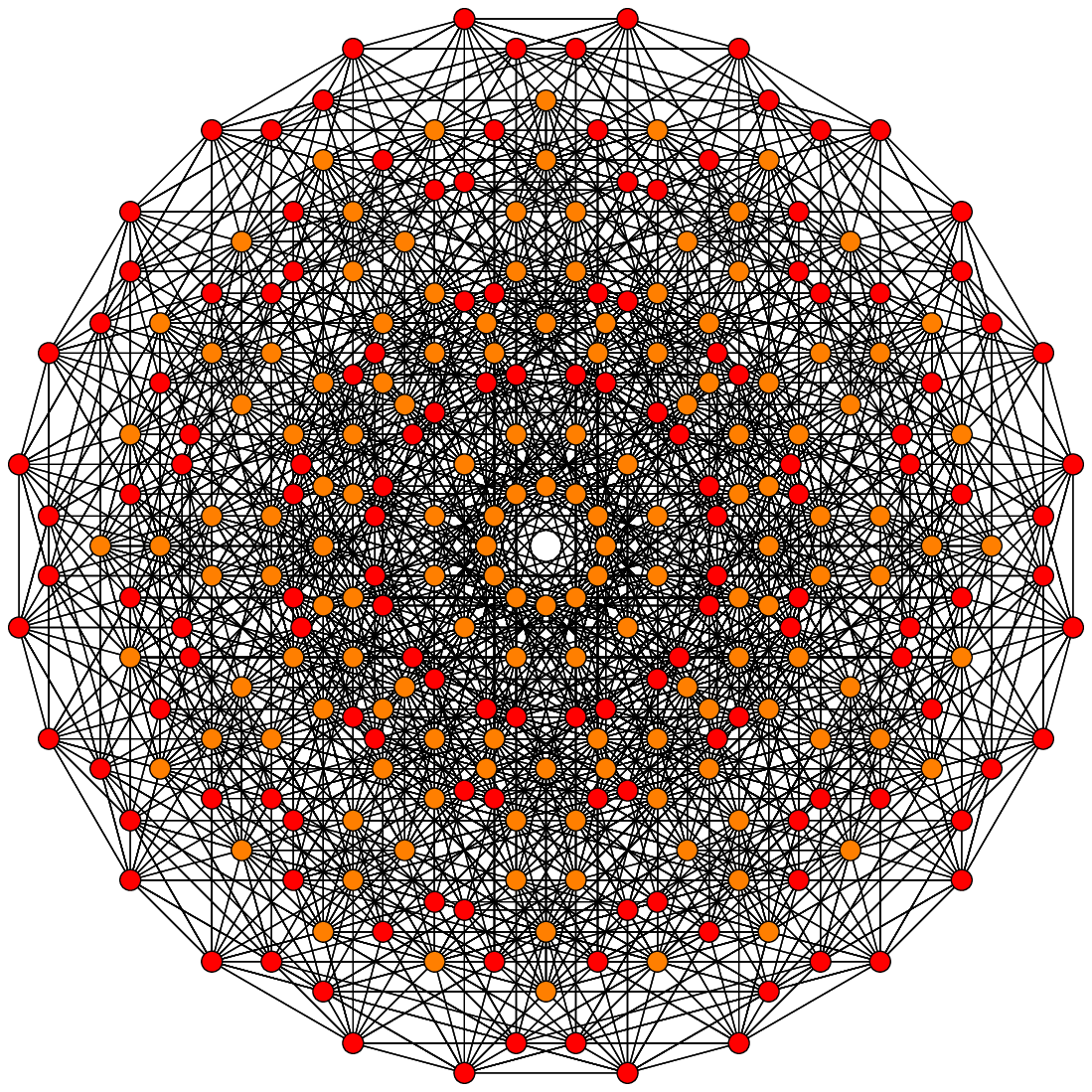}

\end{tabular}\\ \hline
6&\begin{tabular}{c}
$D_{5}$\\
{\bf 144}
\end{tabular}&
\begin{tabular}{c}
\textcolor{black!80}{steric 5-cube or runcinated demipenteract}\\
\\
\scalebox{0.5} 
{
\begin{pspicture}(0,-0.9)(4.02,0.9)
\definecolor{color97b}{rgb}{0.00392156862745098,0.00392156862745098,0.00392156862745098}
\psline[linewidth=0.04cm](1.4,0.7)(1.4,-0.5)
\psline[linewidth=0.04cm](0.2,-0.5)(4.0,-0.5)
\pscircle[linewidth=0.02,dimen=outer,fillstyle=solid,fillcolor=color97b](2.6,-0.5){0.2}
\pscircle[linewidth=0.02,dimen=outer,fillstyle=solid,fillcolor=color97b](1.4,-0.5){0.2}
\pscircle[linewidth=0.02,dimen=outer,fillstyle=solid,fillcolor=color97b](0.2,-0.5){0.2}
\pscircle[linewidth=0.02,dimen=outer,fillstyle=solid,fillcolor=color97b](3.8,-0.5){0.2}
\pscircle[linewidth=0.02,dimen=outer,fillstyle=solid,fillcolor=color97b](1.4,0.7){0.2}
\pscircle[linewidth=0.04,dimen=outer](3.8,-0.5){0.4}
\pscircle[linewidth=0.04,dimen=outer](1.4,0.7){0.4}
\end{pspicture} 
}
\end{tabular}
&$t_{0,3} 1_{21}$&80&
\begin{tabular}{c}
\\ [-0.4cm]
\includegraphics[scale=0.55]{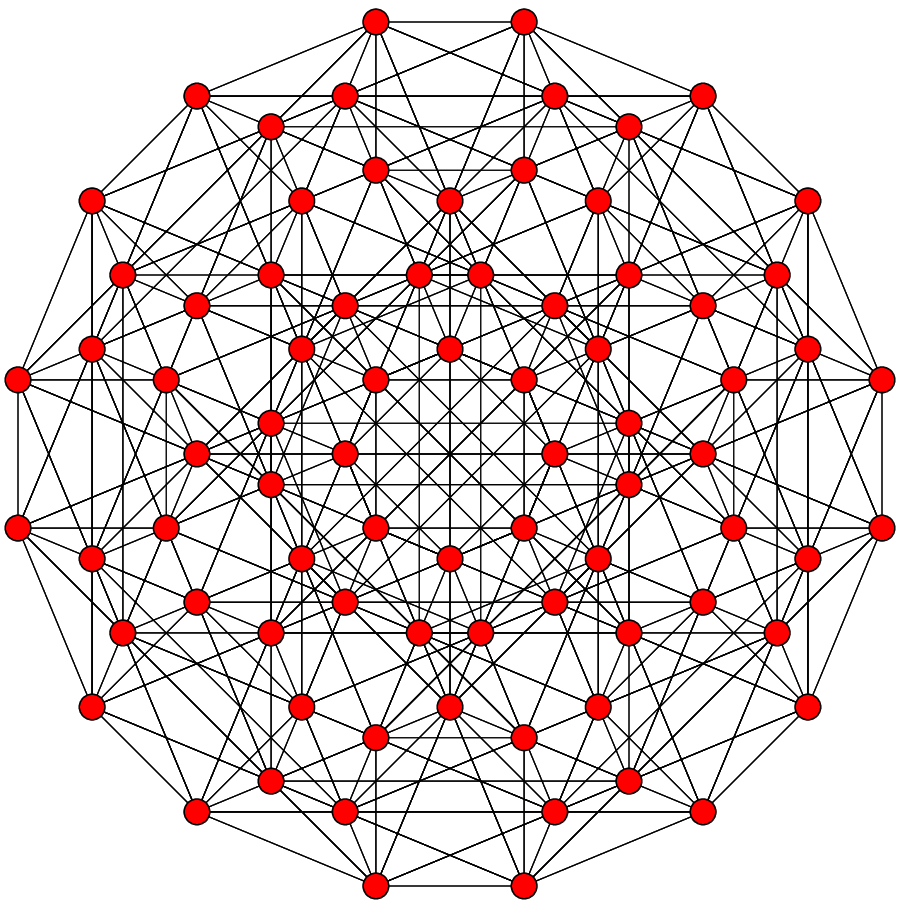}

\end{tabular}\\ \hline
7&\begin{tabular}{c}
$A_{4}$\\
{\bf 24}
\end{tabular}&
\begin{tabular}{c}
\textcolor{black!80}{runcinated 5-cell}\\
\\
\scalebox{0.5} 
{
\begin{pspicture}(0,-0.9)(2.82,0.9)
\definecolor{color97b}{rgb}{0.00392156862745098,0.00392156862745098,0.00392156862745098}
\psline[linewidth=0.04cm](0.2,0.7)(0.2,-0.5)
\psline[linewidth=0.04cm](0.2,-0.5)(2.8,-0.5)
\pscircle[linewidth=0.02,dimen=outer,fillstyle=solid,fillcolor=color97b](1.4,-0.5){0.2}
\pscircle[linewidth=0.02,dimen=outer,fillstyle=solid,fillcolor=color97b](0.2,-0.5){0.2}
\pscircle[linewidth=0.02,dimen=outer,fillstyle=solid,fillcolor=color97b](2.6,-0.5){0.2}
\pscircle[linewidth=0.02,dimen=outer,fillstyle=solid,fillcolor=color97b](0.2,0.7){0.2}
\pscircle[linewidth=0.04,dimen=outer](2.6,-0.5){0.4}
\pscircle[linewidth=0.04,dimen=outer](0.2,0.7){0.4}

\end{pspicture} 
}
\end{tabular}
&$t_{0,3}2_{01}$&20&
\begin{tabular}{c}
\\ [-0.4cm]
\includegraphics[scale=0.38]{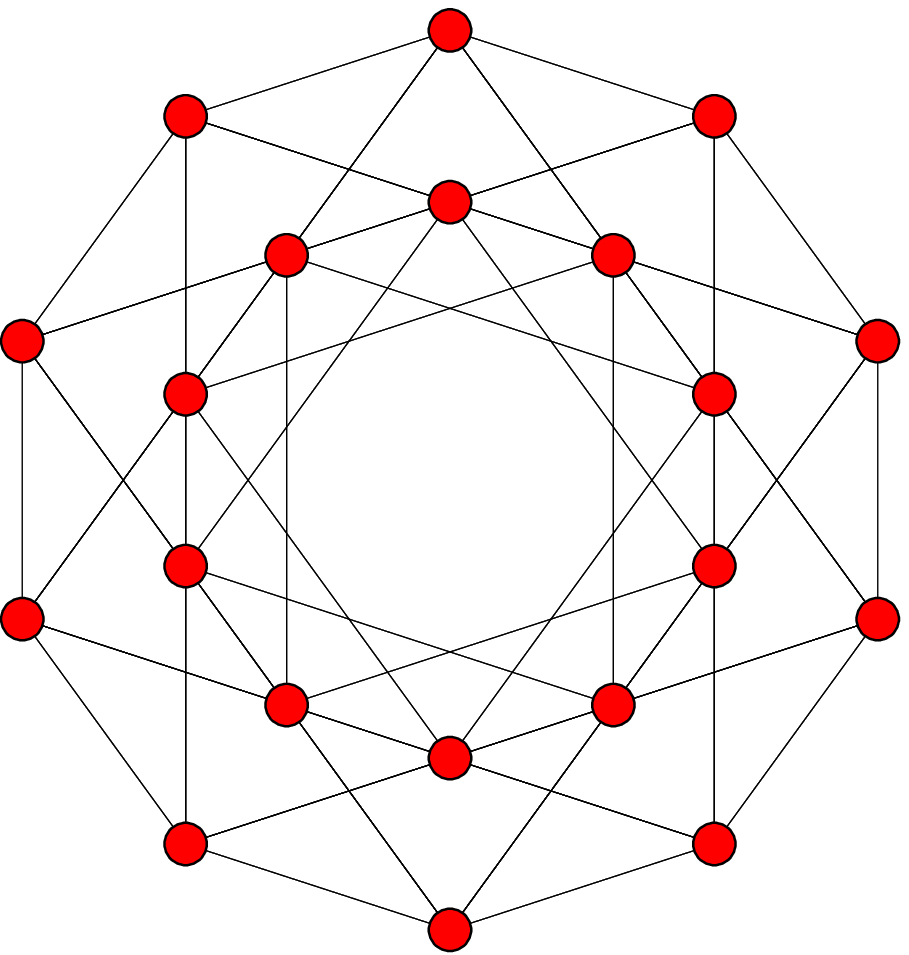}

\end{tabular}
\\ \hline
8&\begin{tabular}{c}
$A_{2}\times A_{1}$\\
($\mathbf{3}$,{\bf 2})
\end{tabular}&
\begin{tabular}{c}
\textcolor{black!80}{triangular prism}\\
\scalebox{0.4} 
{
\begin{pspicture}(0,-0.4)(3.0,0.4)
\definecolor{color458b}{rgb}{0.00392156862745098,0.00392156862745098,0.00392156862745098}
\pscircle[linewidth=0.04,dimen=outer](1.4,0.0){0.4}
\psline[linewidth=0.04cm](0.2,0.0)(1.4,0.0)
\pscircle[linewidth=0.02,dimen=outer,fillstyle=solid,fillcolor=color458b](1.4,0.0){0.2}
\pscircle[linewidth=0.02,dimen=outer,fillstyle=solid,fillcolor=color458b](0.2,0.0){0.2}
\pscircle[linewidth=0.02,dimen=outer,fillstyle=solid,fillcolor=color458b](2.6,0.0){0.2}
\pscircle[linewidth=0.04,dimen=outer](2.6,0.0){0.4}
\end{pspicture} 
}
\end{tabular}

&$-1_{21}$&6&
\begin{tabular}{c}
\\ [-0.4cm]
\includegraphics[scale=0.15]{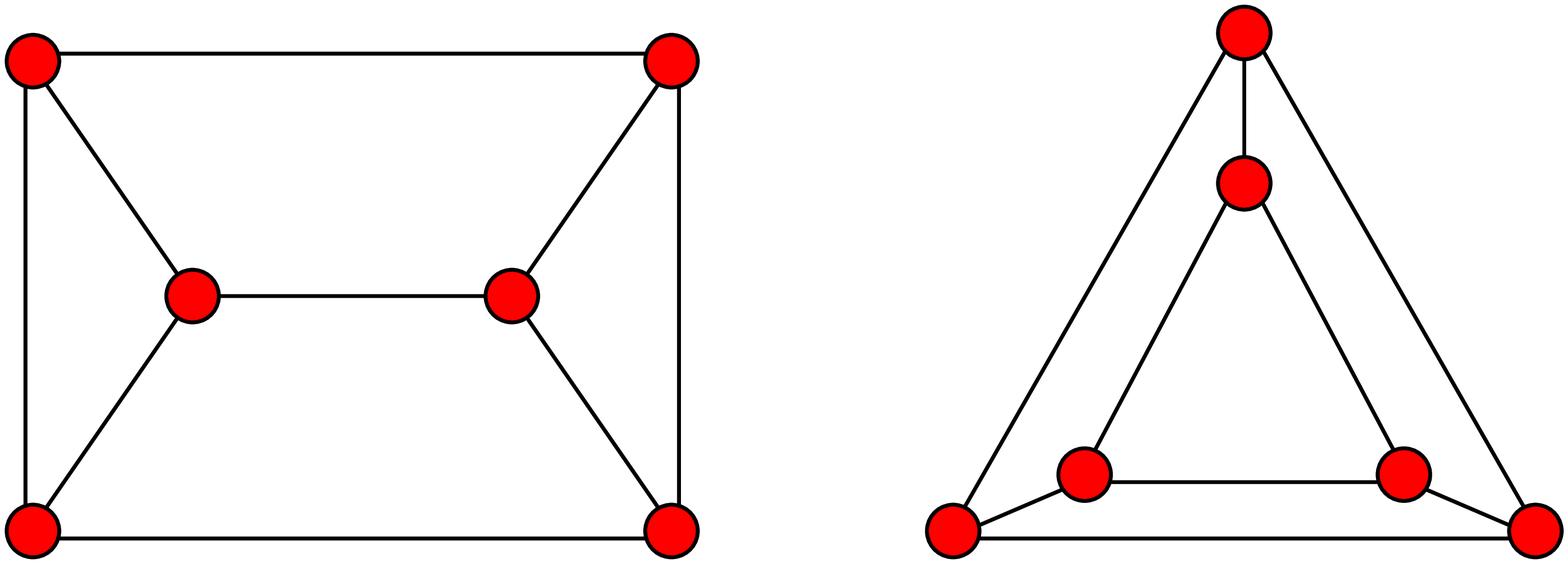}
\end{tabular}
%

\\
\hline

\end{tabular}
}
\end{center}
\caption{We sketch the polytopes associated with 4-branes and their main features, in the last column the corresponding Petrie polygon is drawn with the usual notation.}\label{tab:visual4}
\end{table}
For completeness we list in \cref{tab:polytopeproperties} in \cref{sec:appendix1} all the components of the polytopes we met until now and whose vertices have a correspondence with half-supersymmetric branes in maximal supergravities.

\FloatBarrier

\section*{Conclusions and Perspectives} \addcontentsline{toc}{section}{Conclusions}

Due to their supersymmetry-preserving action, Weyl groups associated with U-duality groups of maximal supergravity theories play a fundamental role in understanding the algebraic structure behind half-supersymmetric branes. An analysis, based on the formalism of reflection groups and Coxeter groups reveals a universal structure behind the hierarchies of 1/2-BPS solutions in maximal theories. This structure is captured by a set of algebraic rules describing the number of independent half-supersymmetric branes, rank by rank, in any dimensions, possessing some striking features. 
The relation between Coxeter group and uniform polytopes provides a new perspective in the analysis of branes: half-supersymmetric branes could be visualized as vertices of certain families of uniform polytopes. From this new perspective it is possible to capture some intriguing properties of and relations between different brane solutions in different theories.\\

In the present paper we analyzed the action of the Weyl group on U-duality representations hosting branes in maximal theories and we discovered a set of algebraic rules describing the number of independent half-supersymmetric solutions. The rules we found,

\begin{subequations}
\begin{flalign}
&N_{0\text{-}brane} =\dfrac{\dim W_{E_{11-d}}}{\dim W_{E_{10-d}}}\\
&N_{1\text{-}brane} =\dfrac{\dim W_{E_{11-d}}}{\dim W_{D_{10-d}}}\\
&N_{2\text{-}brane} =\dfrac{\dim W_{E_{11-d}}}{\dim W_{A_{10-d}}}\\
&N_{3\text{-}brane} =\dfrac{\dim W_{E_{11-d}}}{\dim W_{A_{1}\times A_{9-d}}}\label{eq:3branesrule2}\\
&N_{4\text{-}brane} =\dfrac{\dim W_{E_{11-d}}}{\dim W_{A_{9-d}}} \label{eq:4branesrule2}\\
&N_{5\text{-}brane} =\dfrac{\dim W_{E_{11-d}}}{\dim W_{A_{9-d}}}+\dfrac{\dim W_{E_{11-d}}}{\dim W_{A_{10-d}}}\\
&N_{6\text{-}brane} =\dfrac{\dim W_{E_{11-d}}}{\dim W_{A_{9-d}}}\\
&N_{7\text{-}brane} =\dfrac{\dim W_{E_{11-d}}}{\dim W_{A_{9-d}}}\\ 
&N_{8\text{-}brane} =\dfrac{\dim W_{E_{11-d}}}{\dim W_{A_{9-d}}}. 
\end{flalign}\label{eq:relations1}
\end{subequations}
have some remarkable features. First of all there are different formulae for different rank solutions, in particular we got five types of rules. For p-branes with $p=0,1,2,3,4$ we have five different relations, while, for $p\geqslant 4$, the same relation holds, with an additional contribution for 5-branes described by the same rule appearing for 2-branes. The two contributions for 5-branes were not surprising, since 6-forms live in two different representations corresponding to their coupling with tensor and vector multiplets. Furthermore it is worth noting that the relations above apply both to standard and non-standard branes, revealing that in the full set of U-duality charges, the components coupling to half-supersymmetric solutions follow a well defined pattern.\\

The correspondence between half-supersymmetric solutions and longest weights was a key ingredient in the derivation of our rules and it promotes the Weyl group to the fundamental role it plays in this context. We also remark that, to find the number of independent half-supersymmetric p-branes, the algebraic relations above do not require the knowledge of the representations hosting the corresponding U-duality charges. The formulae \cref{eq:relations1} have the same form despite of the dimension. 
Moreover by inspecting the relations between different rank rules it was possible to uncover some formulae describing the uplift behavior, \cref{eq:upliftsrule}, of half-supersymmetric solutions. All these features characterizing \cref{eq:relations1} make them able to capture a general and deep algebraic structure governing 1/2-BPS branes in maximal theories.\\

Once the rules \cref{eq:relations1} had been found it was natural to look for an interpretation of their coset structure as a symmetry of certain objects. It turned out that these objects are uniform polytopes with the U-duality group as isotropy group and the groups appearing in the denominator of \cref{eq:relations1} as invariance groups of the vertices. This induces a correspondence between branes and vertices of certain families of uniform polytopes providing a new perspective on the hierarchies of half-supersymmetric solutions in maximal theories. In particular we realized that 0-, 1- and 2-branes are in correspondence with the vertices of the families $k_{21}$, $2_{k1}$ and $1_{21}$ of uniform polytopes respectively, with $k$ related to the dimension $d$ by $k=7-d$. 3-branes correspond to rectified $2_{k1}$ polytopes, while for 4-branes the situation is a little less homogeneous since these cannot be identified with a single family of polytopes.\\

The correspondence between half-supersymmetric solutions and polytopes emphasizes another relevant aspect, the relation between rules for different rank solutions. There is a triality relation between 0-, 1- and 2-branes. 0-branes correspond to vertices of the $k_{21}$ polytopes. These polytopes have two types of facets, orthoplexes ans simplexes.
1-branes could be seen as vertices of the polytopes obtained by fixing one vertex on each orthoplex facet, while 2-branes could be seen as vertices of the polytopes obtained by fixing one vertex on each simplex facet. Analogous arguments hold exchanging  the role of the 0-,1- and 2-branes and the corresponding polytopes.\\
Moreover 3-branes correspond to edges of $2_{k1}$ polytopes. For 4-branes and higher rank solutions we found a general behavior. It is manifest by comparison of \cref{eq:4branesrule2} and \cref{eq:3branesrule2} that 4-branes solutions could be obtained from the 3-brane polytopes by adding an orthogonal mirror; this doubles the number of half-supersymmetric 4-brane solutions with respect to 3-branes.
The picture emerging from this description tells us that, the seemingly independent relations for different rank solutions, have quite intriguing links.\\

An immediate application of the correspondence outlined in the present paper is the analysis of the Weyl orbits of less supersymmetric states. These correspond to dominant weights, not highest weights, in the non-standard brane representations we have analyzed.\\

The 0-, 1- and 2-branes in three dimensions maximal supergravity correspond to vertices of the polytopes $4_{21}$, $2_{41}$ and $1_{42}$. The families of uniform polytopes $k_{21}$, $2_{k1}$ have further elements, the honeycombs $5_{21}$, $2_{51}$ corresponding to a symmetry $E_{8}^{+}$. By the same way there is a further honeycomb $6_{21}$ with reflectional symmetry $E_{8}^{++}$.  It would be interesting to look for an extension of the present analysis to two dimensions and one dimension interpreting these honeycombs as the origin of the brane states appearing in maximal supergravities.\\

In perspective it is natural to extend the present work to less supersymmetric theories. In particular these theories are characterized by a U-duality group not appearing in general in its maximal non-compact form. This induces the presence of compact weights. It has been shown that half-supersymmetric solutions correspond to longest non-compact weights \cite{Bergshoeff:2014lxa} thus the analysis of the present work requires a refinement to be applied to the non-maximal cases. This refinement consists in a restriction of the Weyl group to the subgroup generated only by reflections corresponding to non-compact roots.\\

We defined a bridge connecting branes with the world of polytopes. We believe their interplays could provide important improvements in understanding dualities and further clarifying the role that branes play in string theory and supergravity.

%
%
%
%

\newpage
\appendix
\section{Polytopes}\label{sec:appendix1}
In this appendix we report all the components of the polytopes we discuss in \cref{sec:branesandpolytopes}.
\begin{table}[h!]
\renewcommand{\arraystretch}{1.8}
\begin{center}
\resizebox{\textwidth}{!}{
\begin{tabular}{|c|c|c|c|c|c|c|c|c|}
\hline
\bf{Polytope}&{\bf Vertices}&{\bf Edges}&{\bf 2-Faces}&{\bf 3-Faces}&{\bf 4-Faces}&{\bf 5-Faces}&{\bf 6-Faces}&{\bf 7-Faces}\\ \hline \hline
\textcolor{black}{$\mathbf{4_{21}}$}&240&6720&60480&241920&483840&483840&207360&19440\\
\textcolor{black}{$\mathbf{3_{21}}$}&56&756&4032&10080&12096&6048&702&\\
\textcolor{black}{$\mathbf{2_{21}}$}&27&216&720&1080&648&99&&\\
\textcolor{black}{$\mathbf{1_{21}}$}&16&80&160&120&26&&&\\
\textcolor{black}{$\mathbf{0_{21}}$}&10&30&30&10&&&&\\
\textcolor{black}{$\mathbf{-1_{21}}$}&6&9&5&&&&&\\
\textcolor{black}{$\mathbf{2_{41}}$}&2160&69120&483840&1209600&1209600&544320&144960&17520\\
\textcolor{black}{$\mathbf{2_{31}}$}&126&2016&10080&20160&16128&4788&632&\\
\textcolor{black}{$\mathbf{2_{11}}$}&10&40&80&80&32&&&\\
\textcolor{black}{$\mathbf{2_{01}}$}&5&10&10&5&&&&\\
\textcolor{black}{$\mathbf{2_{-11}}$}&3&3&1&&&&&\\
\textcolor{black}{$\mathbf{1_{42}}$}&17280&483840&2419200&3628800&2298240&725760&106080&2400\\
\textcolor{black}{$\mathbf{1_{32}}$}&576&10080&40320&50400&23688&4284&182&\\
\textcolor{black}{$\mathbf{1_{22}}$}&72&720&2160&2160&702&54&&\\
\textcolor{black}{$\mathbf{1_{02}}$}&5&10&10&5&&&&\\
\textcolor{black}{$\mathbf{1_{-12}}$}&2&1&&&&&&\\
\textcolor{black}{$\mathbf{r2_{31}}$}&2016&30240&90720&100800&47880&10332&758&\\
\textcolor{black}{$\mathbf{r2_{21}}$}&216&2160&5040&4320&1350&126&&\\
\textcolor{black}{$\mathbf{r2_{11}}$}&40&240&400&240&42&&&\\
\textcolor{black}{$\mathbf{r2_{01}}$}&10&30&30&10&&&&\\
\textcolor{black}{$\mathbf{r2_{-11}}$}&3&3&1&&&&&\\
\textcolor{black}{{\bf hejack}}&432&3240&7920&7200&2430&342&&\\
\textcolor{black}{{\bf steric 5-cube}}&80&400&720&480&82&&&\\
\textcolor{black}{{\bf runcinated 5-cell}}&20&60&70&30&&&&\\ \hline
\end{tabular}
}
\end{center}
\caption{Components of the uniform polytopes whose vertices could be associated with half-supersymmetric solutions in maximal theories.}\label{tab:polytopeproperties}
\end{table}

\section{Petrie Polygons}\label{sec:appB}
In this section we review the construction of the Petrie polygons appearing in the paper. A Petrie polygon of an n-polytope is a polygon such that every consecutive n-1 edges, but not n belong to the same facet of the polytope \cite{coxeter1973regular}. For a given polytope the Petrie Polygon could be obtained as projection on the Coxeter plane.  The Coxeter plane is defined by the action of a Coxeter element $w$ as the plane on which it acts as a rotation of $2\pi/h$, where $h$ is the Coxeter number, i.e. the order of the Coxeter elements (we recall that Coxeter elements are all conjugate). Taking a Coxeter element  $w$ in a Coxeter system $(W,S)$ with Coxeter number $h$ it has eigenvalues
\begin{flalign}
 &\lambda_{i}=e^{2ik\pi/h},
\end{flalign}
for some $k\in \mathds{Z}$. If we call $z_{k}\in \mathds{C}^{n}$ its eigenvectors then we can write
\begin{flalign}
 &wz_{k}=e^{2ik\pi/h}z_{k}.
\end{flalign}
$w$ acts as rotation of $2k\pi/h$ on $z_{k}$. The Coxeter plane is identified by the element $z_{1}$, always appearing in the set of eigenvectors.  We decomposed $z_{1}$ in its real and imaginary parts 
\begin{flalign}
 &z_1=\Re z_1+i\Im z_1
\end{flalign}
and we consider the plane $\{\Re z_1,\Im z_1\}$, where $\Re z_1,\Im z_1\in \mathds{R}^{n}$. Thus given a weight $\Lambda$ its projection on the Coxeter plane has components 
\begin{flalign}
 &P_{\Lambda}=\biggl(\scp{\Lambda}{\Re z_1},\    \scp{\Lambda}{\Im z_1} \biggr).
\end{flalign}
We could discuss a simple example; let's take the representation {\bf 10} of $D_{5}$.  $D_5$, whose Dynkin diagram is in \cref{fig:D5petrieexample}, has Coxeter number $h=8$.
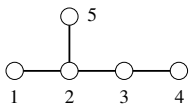
\begin{figure}[h]
 \centering
\scalebox{0.6} 
{
\begin{pspicture}(0,-1.1034375)(4.02,1.1034375)
\definecolor{color6b}{rgb}{0.996078431372549,0.996078431372549,0.996078431372549}
\psline[linewidth=0.04cm](1.4,0.8428125)(1.4,-0.3571875)
\psline[linewidth=0.04cm](0.2,-0.3571875)(4.0,-0.3571875)
\pscircle[linewidth=0.02,dimen=outer,fillstyle=solid,fillcolor=color6b](2.6,-0.3571875){0.2}
\pscircle[linewidth=0.02,dimen=outer,fillstyle=solid,fillcolor=color6b](1.4,-0.3571875){0.2}
\pscircle[linewidth=0.02,dimen=outer,fillstyle=solid,fillcolor=color6b](0.2,-0.3571875){0.2}
\pscircle[linewidth=0.02,dimen=outer,fillstyle=solid,fillcolor=color6b](3.8,-0.3571875){0.2}
\pscircle[linewidth=0.02,dimen=outer,fillstyle=solid,fillcolor=color6b](1.4,0.8428125){0.2}
\usefont{T1}{ptm}{m}{n}
\rput(0.19546875,-0.9221875){\large 1}
\usefont{T1}{ptm}{m}{n}
\rput(1.90375,0.8778125){\large 5}
\usefont{T1}{ptm}{m}{n}
\rput(3.816875,-0.9221875){\large 4}
\usefont{T1}{ptm}{m}{n}
\rput(2.605625,-0.9221875){\large 3}
\usefont{T1}{ptm}{m}{n}
\rput(1.401875,-0.9221875){\large 2}
\end{pspicture} 
}
\caption{$D_5$ Dynkin diagram.}\label{fig:D5petrieexample}
\end{figure}\\
We choose as Coxeter element 
\begin{flalign}
 &w=w_{5}w_{1}w_{3}w_{4}w_{2}.
\end{flalign}
Among the Coxeter elements the one we have chosen is called \textbf{distinguished Coxeter element} since it is the product of two involutions, $r_1=w_{5}w_{1}w_{3}$ and $r_2=w_{4}w_{2}$ with elements commuting each others. Its action on weight vectors could be represented by the matrix
\begin{flalign}
 &M_{w}=\left(\begin{array}{ccccc}
         0& -1& 1& 0& 1\\
         1& -1& 1& 0& 1\\
         1& -1& 1& -1& 1\\
         0& 0& 1& -1&0\\ 
         1& -1& 1& 0& 0
        \end{array}\right),
\end{flalign}
with eigenvalues
\begin{flalign}
 &\lambda_{k}=e^{ik\pi/4}\qquad\text{  for  }\ k=1,3,4,5,7.
\end{flalign}
The eigenvector corresponding to $\lambda_{1}$ is 
\begin{flalign}
 &z_{1}=\left(1,\ 1+e^{i7\pi/4},\ \sqrt{2},\ -i+e^{i\pi/4},\ 1\right)
\end{flalign}
and the Coxeter plane is identified by the vectors
\begin{subequations}
\begin{flalign}
 &\Re z_{1}=\left(1,\ 1+\cos(7\pi/4),\ \sqrt{2},\ -\cos(\pi/4),\ 1\right)\\
 &\Im z_{1}=\left(0,\ \sin(7\pi/4),\ 0,\ \sin(\pi/4)-1,\ 0\right).
\end{flalign}\label{eq:Coxeterplanereim1}
\end{subequations}
The weights of the representation {\bf 10} of $D_{5}$ appear in its Dynkin tree in \cref{fig:Dynkintree10odD5}. They correspond to a 5-orthoplex.
\FloatBarrier
\begin{figure}[h!]
\centering
\scalebox{0.7} 
{
\begin{pspicture}(0,-7.5673957)(8.560104,7.5673957)
\psline[linewidth=0.04cm,dotsize=0.07055555cm 2.0]{*-*}(1.1930208,-0.0746875)(4.193021,-1.8746876)
\usefont{T1}{ptm}{m}{n}
\rput{-31.396902}(0.8713021,1.3813276){\rput(2.8564584,-0.8396875){\large $\alpha_1$}}
\psline[linewidth=0.04cm,dotsize=0.07055555cm 2.0]{*-*}(4.193021,-1.8746876)(5.493021,-3.6746874)
\usefont{T1}{ptm}{m}{n}
\rput{-54.04724}(4.214296,2.9438493){\rput(4.9564586,-2.6396875){\large $\alpha_2$}}
\psline[linewidth=0.04cm,dotsize=0.07055555cm 2.0]{*-*}(4.193021,1.7253125)(1.1930208,-0.0746875)
\usefont{T1}{ptm}{m}{n}
\rput{31.135086}(0.925991,-1.2425455){\rput(2.6564584,1.0603125){\large $\alpha_5$}}
\psline[linewidth=0.04cm,dotsize=0.07055555cm 2.0]{*-*}(2.8930209,3.5253124)(4.193021,1.7253125)
\usefont{T1}{ptm}{m}{n}
\rput{-54.04724}(-0.69375783,4.12106){\rput(3.6564584,2.7603126){\large $\alpha_2$}}
\psline[linewidth=0.04cm,dotsize=0.07055555cm 2.0]{*-*}(4.193021,1.7253125)(7.193021,-0.0746875)
\usefont{T1}{ptm}{m}{n}
\rput{-31.396902}(0.37283093,3.2077758){\rput(5.856458,0.9603125){\large $\alpha_1$}}
\psline[linewidth=0.04cm,dotsize=0.07055555cm 2.0]{*-*}(7.193021,-0.0746875)(4.193021,-1.8746876)
\usefont{T1}{ptm}{m}{n}
\rput{31.135086}(0.42743552,-3.0530071){\rput(5.6564584,-0.7396875){\large $\alpha_5$}}
\psline[linewidth=0.04cm,dotsize=0.07055555cm 2.0]{*-*}(5.493021,-3.6746874)(5.493021,-5.4746876)
\usefont{T1}{ptm}{m}{n}
\rput{-90.0}(10.152395,1.2336459){\rput(5.6564584,-4.4396877){\large $\alpha_3$}}
\psline[linewidth=0.04cm,dotsize=0.07055555cm 2.0]{*-*}(5.493021,-5.4746876)(4.193021,-7.2746873)
\usefont{T1}{ptm}{m}{n}
\rput{55.263065}(-3.1249163,-6.5493407){\rput(4.6564584,-6.2396874){\large $\alpha_4$}}
\psline[linewidth=0.04cm,dotsize=0.07055555cm 2.0]{*-*}(4.193021,7.1253123)(2.8930209,5.3253126)
\usefont{T1}{ptm}{m}{n}
\rput{55.263065}(6.6702247,-0.06062886){\rput(3.3564584,6.3603125){\large $\alpha_4$}}
\psline[linewidth=0.04cm,dotsize=0.07055555cm 2.0]{*-*}(2.8930209,5.3253126)(2.8930209,3.5253124)
\usefont{T1}{ptm}{m}{n}
\rput{-90.0}(-1.4476042,7.633646){\rput(3.0564582,4.5603123){\large $\alpha_3$}}
\usefont{T1}{ptm}{m}{n}
\rput(2.9239583,3.4703126){\Large \psframebox[linewidth=0.04,fillstyle=solid]{0 1 -1 0 0}}
\usefont{T1}{ptm}{m}{n}
\rput(2.9239583,5.2703123){\Large \psframebox[linewidth=0.04,fillstyle=solid]{0 0 1 -1 0}}
\usefont{T1}{ptm}{m}{n}
\rput(4.1339583,7.1703124){\Large \psframebox[linewidth=0.04,fillstyle=solid]{0 0 0 1 0}}
\usefont{T1}{ptm}{m}{n}
\rput(7.221927,-0.1296875){\Large \psframebox[linewidth=0.04,fillstyle=solid]{-1 0 0 0 1}}
\usefont{T1}{ptm}{m}{n}
\rput(1.2072396,-0.1296875){\Large \psframebox[linewidth=0.04,fillstyle=solid]{1 0 0 0 -1}}
\usefont{T1}{ptm}{m}{n}
\rput(4.1072397,1.6703125){\Large \psframebox[linewidth=0.04,fillstyle=solid]{1 -1 0 0 1}}
\usefont{T1}{ptm}{m}{n}
\rput(4.123958,-7.2296877){\Large \psframebox[linewidth=0.04,fillstyle=solid]{0 0 0 -1 0}}
\usefont{T1}{ptm}{m}{n}
\rput(5.4239583,-5.5296874){\Large \psframebox[linewidth=0.04,fillstyle=solid]{0 0 -1 1 0}}
\usefont{T1}{ptm}{m}{n}
\rput(5.4239583,-3.7296875){\Large \psframebox[linewidth=0.04,fillstyle=solid]{0 -1 1 0 0}}
\usefont{T1}{ptm}{m}{n}
\rput(4.211927,-1.9296875){\Large \psframebox[linewidth=0.04,fillstyle=solid]{-1 1 0 0 -1}}
\end{pspicture} 
}

\caption{Dynkin tree of the representation {\bf 10} of $D_{5}$.} \label{fig:Dynkintree10odD5}
\end{figure}
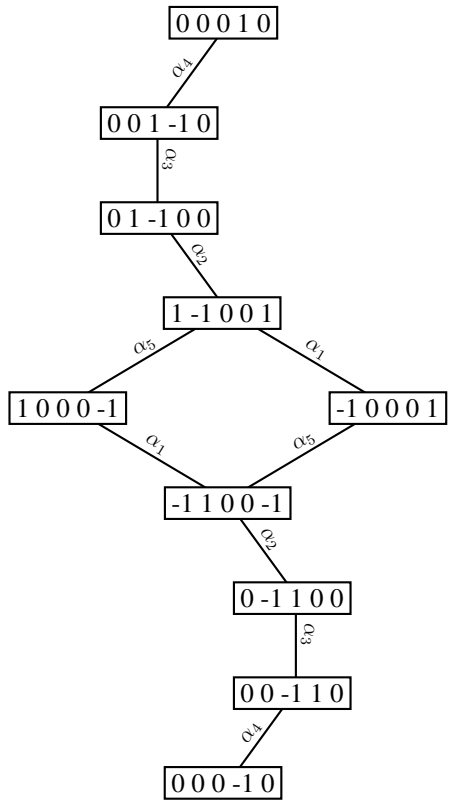
\FloatBarrier
Since the vectors in \cref{eq:Coxeterplanereim1} have coordinates in the basis of simple roots, the projection could be realized just taking their Euclidean product with the vector of Dynkin labels of the weights. The two weights $\boxed{\pm 1\ 0\ 0\ 0\ \mp 1}$ are projected to $(0,0)$  on the $D_5$ Coxeter plane, while the other weights have projection corresponding to the vertices of a regular octagon as in \cref{fig:10ofD5octagon}.  With the same notation describes previously, red points have no degeneracy while the orange point is doubly degenerate.
\FloatBarrier
\begin{figure}[h!]
\centering
\includegraphics[scale=0.5]{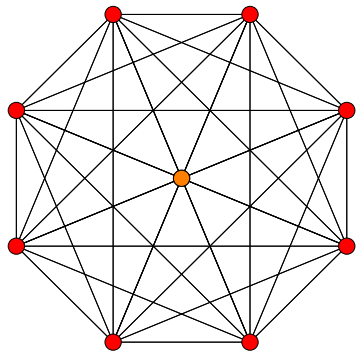}
\caption{Petrie Polygon of the representation {\bf 10} of $D_{5}$ corresponding to a 5-orthoplex.} \label{fig:10ofD5octagon}
\end{figure}
Petrie polygons are rather useful in studying the properties of higher dimensional polytopes.

\FloatBarrier

\nocite{}
 \bibliography{bibliography}{}
\end{document}